%% file: main.tex
\documentclass[12pt,english]{article}

\input{preamble.tex}
\title{Nonparametric inference on counterfactuals\\ in first-price auctions\footnote{This paper supersedes the earlier version entitled ``Simple nonparametric inference for first-price auctions via bid spacings'' available at \href{https://arxiv.org/pdf/2106.13856v1.pdf}{https://arxiv.org/pdf/2106.13856v1.pdf}. A Python package implementing the methodology in the paper is available at \href{https://pypi.org/project/simple-fpa}{https://pypi.org/project/simple-fpa}}}
\author{
\setcounter{footnote}{1}
    Pasha Andreyanov\thanks{
	Faculty of Economic Sciences, HSE University.
Email: \href{pandreyanov@gmail.com}{pandreyanov@gmail.com}
	}
	\and Grigory Franguridi\thanks{
		Center for Economic and Social Research, University of Southern California.
		Email: \href{franguri@usc.edu}{franguri@usc.edu}
	}
}
\date{\today}

\begin{document}

\maketitle

\input{abstract}

\input{intro}

\input{setup}

\input{valuation_quantiles}

\input{functionals}

\input{MC}

\input{empirical}

\input{practical}

\input{conclusion}

\newpage  
\part*{Appendix}
\appendix

\input{appendix}

\newpage  
\bibliographystyle{ecta}
\bibliography{main}

\end{document}

%% file: preamble.tex
\usepackage[T1]{fontenc}
\usepackage[margin=1in]{geometry}
\usepackage[utf8]{inputenc}
\usepackage{color}
\usepackage{mathtools}

\usepackage[usenames,dvipsnames]{xcolor}

\usepackage{amsfonts, amssymb, amsmath, amsthm}

\usepackage{hyperref}
\definecolor{darkblue}{rgb}{0, 0, 0.5}
\hypersetup{
     colorlinks = true,
     citecolor = darkblue,
     urlcolor = darkblue,
     linkcolor = darkblue
}
\usepackage[all]{hypcap}
\usepackage[nameinlink, capitalize, noabbrev]{cleveref}

\usepackage[dvipsnames]{xcolor}
\usepackage{listings}

\definecolor{codegreen}{rgb}{0,0.6,0}
\definecolor{codegray}{rgb}{0.5,0.5,0.5}
\definecolor{codepurple}{rgb}{0.58,0,0.82}
\definecolor{backcolour}{rgb}{0.95,0.95,0.92}
\definecolor{deepblue}{rgb}{0,0,.9}

\lstdefinestyle{mystyle}{
  language = Python,
  otherkeywords={@},
  deletekeywords={filter},
  deletendkeywords={max, min, abs},
  backgroundcolor=\color{white}, commentstyle=\color{codegreen},
  keywordstyle=\color{magenta},
  numberstyle=\tiny\color{codegray},
  stringstyle=\color{codepurple},
  basicstyle=\ttfamily\footnotesize,
  breakatwhitespace=false,         
  breaklines=true,                 
  captionpos=b,                    
  keepspaces=true,                 
  numbers=left,                    
  numbersep=5pt,                  
  showspaces=false,                
  showstringspaces=false,
  showtabs=false,                  
  tabsize=2,
  emph={zeros, prange, jit, power, std, rvs, maximum, roll, convolve, array, max, min, pdf, ppf, uniform, random, sqrt, linspace, sort, mean, percentile, sign, abs, powerlaw, beta, stats},
  emphstyle=\color{deepblue}
}

\usepackage{changes}
\definechangesauthor[name=Pasha, color=BrickRed]{P}
\definechangesauthor[name=Grisha, color=ForestGreen]{G}

\newcommand{\boldred}[1]{\boldsymbol{\textcolor{BrickRed}{#1}}}

\usepackage{enumerate}
\usepackage{enumitem} 
\usepackage{cases}
\usepackage{bm}
\usepackage[authoryear]{natbib}
\usepackage{setspace}
\usepackage{float}

\usepackage{graphicx, graphics}

\usepackage{booktabs}
\usepackage{rotating}
\usepackage[flushleft]{threeparttable}
\usepackage{caption, subcaption}
\usepackage{here, float}
\usepackage{tabularx, makecell}
\usepackage{adjustbox}
\usepackage{tabu}

\usepackage{macros}

\makeatletter
\newtheorem{rem}{Remark}

\usepackage[ruled, vlined]{algorithm2e}

\newtheorem{assumption}{Assumption}
\crefname{assumption}{Assumption}{Assumptions}
\crefname{prop}{Proposition}{Propositions}
\crefname{lem}{Lemma}{Lemmas}
\crefname{thm}{Theorem}{Theorems}

\newtheorem{cor}{Corollary}
\newtheorem{lem}{Lemma}
\newtheorem{thm}{Theorem}
\newtheorem{prop}{Proposition}
\newtheorem*{thm*}{Theorem}

%% file: abstract.tex
\begin{abstract}

In a classical model of the first-price sealed-bid auction with independent private values, we develop nonparametric estimators for several policy-relevant targets, such as the bidder's surplus and auctioneer's revenue under counterfactual reserve prices. Motivated by the linearity of these targets in the quantile function of bidders' values, we propose an estimator of the latter and derive its Bahadur-Kiefer expansion.
This makes it possible to construct exact uniform
confidence bands and test complex hypotheses about the auction design.
Using the data on U.S. Forest Service timber auctions, we test whether setting zero reserve prices in these auctions was revenue maximizing.

\medskip

\noindent \textbf{JEL Classfication:} C57, D44

\medskip


\end{abstract}
\newpage

%% file: intro.tex
\section{Introduction}
\label{sec:intro}

In the empirical studies of first-price auctions, a structural approach to estimation and inference is often used. This approach exploits restrictions derived from economic theory to recover bidders' latent valuations from the observed bids. With these valuations, the researcher can make predictions about the effects of changes in auction rules or the composition of bidders. 
Various methods, in both parametric and nonparametric frameworks, have been developed; see, e.g., \cite{paarsch2006introduction}, \cite{athey2007nonparametric}, and \cite{perrigne2019econometrics} for an overview. 

Since the seminal papers by \cite{elyakime1994first}, \cite{guerre2000optimal} and \cite{li2000conditionally}, it is the probability density function (PDF) of bidders' values that has been considered a default target of nonparametric analysis.
This choice is natural since it allows for constructive identification \citep{matzkin2013nonparametric} when valuations are independent\footnote{With correlated valuations, nonparametric identification is partial, see \cite{aradillas2013identification}.}.
However, researchers are often interested in revenue and other targets, which are \emph{nonlinear} functionals of the density.
Even with the asymptotic theory for the density estimator developed, constructing the confidence intervals and bands for these targets is not simple.
This often leads researchers to report confidence intervals based on simulation from the estimated PDF \citep[e.g.,][]{li2003timber} or none at all. 

Our primary focus is testing complex hypotheses about several natural targets, such as total surplus, bidders' surplus, and auctioneer's revenue. For example, before introducing a reserve price into the auction, a policy-maker might want to test whether \emph{any} reserve price would yield at least a 5\% increase in the auctioneer's revenue or the median bidder's surplus.\footnote{The auctioneer might be concerned about the bidder's welfare as much as about his own revenue, see, e.g., \cite{andreyanov2023past}. The auctioneer's expected revenue and the median bidder's surplus are typically maximized at a positive reserve price.} 
 

To capture a wide variety of data-generating processes, we allow for binding reserve prices, risk aversion in bidders' preferences, and random (i.e., unknown) number of bidders -- a delicate feature often neglected in the literature.
Indeed, many auctions, including procurement auctions, are sealed-bid with bids submitted online, so the participants do not know the number of active bidders.

Our approach relies on the \textit{quantile function of valuations} --- an alternative candidate for constructive identification --- and on the observation that many targets, including revenue, are \emph{continuous linear functionals} of this quantile function.
Since the value quantile function is the key ingredient of our counterfactual evaluation, we provide its complete first-order asymptotic analysis.
Namely, we derive the uniform, asymptotically linear (Bahadur-Kiefer, or BK) expansion for the kernel estimator $\hat v_h$ of the value quantile function $v$, where $h$ is a smoothing bandwidth.
This expansion implies that, despite converging to a Gaussian distribution pointwise, the estimator does not admit a functional central limit theorem, which calls for alternative ways of conducting uniform inference.
Luckily, the linear term of the studentized estimator is \emph{known} and \emph{pivotal}, allowing us to suggest simple simulation-based confidence bands -- a viable alternative to bootstrap -- and establish their validity using the anti-concentration theory of \citet{chernozhukov2014gaussian}.

With the asymptotic theory of the value quantile function at hand, we move towards analyzing the estimators of our targets at counterfactual reserve prices.
We show that these can be divided into two broad classes.
One class contains the ``smoother'' (w.r.t. the value quantile function) counterfactuals that are estimable at the parametric rate $n^{-1/2}$ and converge weakly to a Gaussian process in $\linfzeroone$. The other class contains the ``less smooth'' counterfactuals that are only estimable at the slower rate $(nh)^{-1/2}$ and do not converge weakly in $\linfzeroone$.
For each class, we develop a distinct protocol for constructing confidence intervals and bands and establish their validity.
Under the appropriate choice of bandwidth, the convergence rate of the ``less smooth'' estimator is MSE-optimal. On the other hand, for inference, undersmoothing needs to be used.\footnote{The undersmoothing approach is standard in this literature, see, e.g., Assumption 4.1 in \citet{zincenko2024estimation} and Equation (5.11) in \citet{gimenes2021quantile}.}

To illustrate our methodology, we use Philip Haile's data on U.S. Forest Service timber auctions, where a reserve price was set to zero, and assess the optimality of this auction design. Namely, we test whether the seller's expected revenue could have increased if the auction designer chose a nonzero reserve price.

While the basic properties of the classical two-step estimator of the density of latent valuations were stated in \cite{guerre2000optimal}, an exhaustive theoretical analysis was completed much later in \cite{ma2019inference}. However, even this analysis did not automatically extend to the functionals of interest. The complete first-order analysis of a single functional -- revenue -- was performed only in \cite{zincenko2024estimation} and required a significant effort. At the same time, a competing, quantile-regression-based approach was developed in \cite{guerre2012uniform}, allowing for a wide range of targets. However, their analysis is missing uniformity over quantiles (i.e., counterfactual reserve prices), which is needed to test the kinds of hypotheses that are the focus of our study. Thus, our work is most similar to \cite{zincenko2024estimation} in spirit, but we also consider targets other than revenue.

More broadly, our work contributes to the expanding literature on quantile methods in first-price auctions, see \cite{marmer2012quantile} and \cite{enache2017quantile} for kernel-based estimators, \cite{luo2018integrated} for isotone regression-based estimators, and \cite{guerre2012uniform} and \cite{gimenes2021quantile} for local polynomial estimators. Interestingly, our estimator of valuation quantiles is a weighted sum of the differences of ordered bids, often referred to as \emph{bid spacings}. The latter has been used for collision detection in \cite{ingraham2005test}, for set identification of bidders' rents in \cite{PAUL2004103} and \cite{marra2020sample}, and in the prior-free clock auction design in \cite{loer2020spacings}.

The rest of the paper is organized as follows.
In \cref{sec:framework}, we set up the theoretical and econometric framework for our analysis.
In \cref{sec:BK} and \cref{sec:func}, we develop estimation and inference procedures for the value quantile function and the counterfactual quantities of interest, respectively.
In \cref{sec:MC}, we provide the Monte Carlo simulations of the finite-sample coverage of our confidence bands.
In \cref{sec:emp}, we use the timber auction data to test whether counterfactual reserve prices increase the auctioneer's revenue. 
\cref{sec:practical} contains a discussion of some practical aspects of our methodology.
\cref{sec:conclude} concludes the paper.
Proofs of theoretical results are provided in the Appendix.

%% file: setup.tex
\section{Framework}
\label{sec:framework}

\subsection{Baseline model}

We seek a symmetric Bayes-Nash equilibrium in a first-price auction with $M \ge 2$ ex-ante identical and risk-neutral bidders. Denote the valuation of a bidder by $v$. We impose the following assumption on the value distribution \citep[][Definition 2]{guerre2009nonparametric}.
\begin{assumption}[Distribution of values]\label{ass:distr-bids}
The values $v_1,\dots,v_M$ of bidders are drawn independently from a common CDF $G$ with support $[0, \bar v]$ that is twice continuously differentiable and has a strictly positive density $g(v) = G'(v)$ for all $v \in [0,\bar v]$.
\end{assumption}
The auctioneer announces a reserve price $r^*>0$, which is necessarily binding. Every bidder then submits a sealed bid of $b$. The equilibrium bidding strategy $\beta(v)$ can be characterized by applying the Envelope Theorem to the maximization of $(b-v)G^{M-1}(\beta^{-1}(b|r^*))$:
\begin{equation}\label{eq:envelope}
    \beta(v|r^*) = v - \frac{\int_{r^*}^v G^{M-1}(x)\,dx}{G^{M-1}(v)} = \frac{G^{M-1}(r^*)r^* + \int_{r^*}^v x d G^{M-1}(x)}{G^{M-1}(v)},
\end{equation}
for all $v \geqslant r^* \geqslant 0$ see, e.g., \cite{rileysamuelson} or \cite{krishna2009auction}. 

This strategy is strictly increasing and twice continuously differentiable. Moreover, by the first-order conditions\footnote{$\beta'(v|r^*) = \frac{(v-\beta(v|r^*))g(v)}{G(v)/(M-1)}>0$ for all $v > r^*$ and $\beta'(r^*|r^*) = \frac{1}{1+(G/g)'(r^*)/(M-1)} > 0$ by L'H\^{o}pital's rule.}, it has a strictly positive derivative on $[0,\bar v]$. Denoting by $F$ the CDF of the equilibrium bid and $f=F'$, the inverse bidding strategy can be written as
\begin{equation}
    v = \beta^{-1}(b|r^*) = b + \frac{F(b)}{(M-1)f(b)},
    \label{iden1}
\end{equation}
allowing the recovery of the latent values from the observed bids. This suggests a nonparametric estimation approach popularized by \citet{guerre2000optimal} and  \cite{li2000conditionally}.

Alternatively, we can rewrite the equation \eqref{iden1} in terms of the quantiles of the participating values. Denote by $Q(u) \bydef F^{-1}(u)$ the \emph{bid quantile function} and by $q(u)\bydef Q'(u)$ the associated \emph{bid quantile density}. Let $v(u)\bydef G^{-1}(u)$ be the $u$-th quantile of the participating value distribution $G$. Then equation \eqref{iden1} can be rewritten as
\begin{align}
    v(u) = \beta^{-1}(Q(u)) = Q(u) + \frac{u q(u)}{M-1},
\end{align}
where we use the change of variables $b=Q(u)$ along with the identities $F(Q(u))=u$ and $f(Q(u))q(u) = 1$. Since, by definition, $Q(u) = \beta(G^{-1}(u)|r^*)$, and both $(G^{-1})'(v)$ and $\beta'(v|r^*)$ are strictly positive for all $v \in [r^*, \bar v]$, we arrive at the following property.

\begin{prop}
[Distribution of bids]
\label{prop:distr-of-bids}
Under \cref{ass:distr-bids}, the equilibrium bids are drawn independently from a distribution with a twice continuously differentiable quantile function $Q$ such that $q(u) = Q'(u) > 0$ for all $u \in [0,1]$.
\end{prop}

\subsection{Counterfactuals}\label{sec:counterfactuals}

We show that a variety of counterfactual metrics can be written in terms of the counterfactual reserve price $r^*$ and the distribution $G$ of bids submitted under the \emph{original} reserve price $\underline r$. We then show that, in our model, these counterfactual metrics can be rewritten as \emph{linear} functionals of the value quantile function $v(\cdot)$, the key observation enabling simple inference procedures in \cref{sec:func}.


One such counterfactual is the total expected (ex ante) surplus. In a symmetric equilibrium, it is ex post equal to the highest valuation if it exceeds $r^*$, and zero otherwise, which is a random variable with CDF $G^M(\cdot)$. Hence the total surplus is its expectation
\begin{equation}
\textit{TS}(r^*) \bydef \int_{r^*}^{\bar v} {x} d G^M(x) .
\end{equation}

Another counterfactual is bidder's expected surplus. By the revenue equivalence principle \citep[see][]{krishna2009auction},
 the interim surplus of a bidder is related to her equilibrium probability of winning, equal to $G^{M-1}(v)$, via the envelope conditions
\begin{equation}
\pi(v|r^*) \bydef \int_{r^{\ast}}^{v} G^{M-1}(x) \, d x.
\end{equation}
To derive the bidder's expected (ex ante) surplus $\textit{BS}$, we need to take the expectation of $\pi(v|r^*)$ w.r.t. the distribution of $v$. Integration by parts yields the formula
\begin{equation*}
\textit{BS}(r^*) \bydef \int_{r^{\ast}}^{\bar v} \pi(x|r^*) d G(x) = \int_{r^{\ast}}^{\bar v}  \int_{r^{\ast}}^v G^{M-1}(x) d x d(G(v)-1) = \int_{r^{\ast}}^{\bar v} (1-G(x))G^{M-1}(x) dx.
\end{equation*}

Finally, we consider the seller's expected revenue under the counterfactual reserve price $r^*$, which is equal to the difference between the total expected surplus and $M$ times the bidder's expected surplus,
\begin{align*}
&\textit{RE}(r^*) \bydef \textit{TS}(r^*) - M \cdot \textit{BS}(r^*) = \int_{r^{\ast}}^{\bar v} x d \left(G^M(x) \right) - M \int_{r^{\ast}}^{\bar v} (1-G(x))G^{M-1}(x) d x = \\
&= M (1-G(r^*))G^{M-1}(r^*)r^{\ast} + \int_{r^{\ast}}^{\bar v} x d \left( G^M(x)+ M (1-G(x))G^{M-1}(x))\right).
\end{align*}

\begin{table}[t!]
\begin{tabularx}{\textwidth}{l|X|X}
\toprule
\toprule
 & Classical (nonlinear) form & Quantile (linear in $v$) form\\
\midrule
\makecell[l]{\textit{TS}$(r^*)$} & $\int_{r^*}^{\bar v} x \, d G^M(x)$ & $M\int_{u^{\ast}}^{1} z^{M-1}\boldred{v(z)} \, dz$\\
\midrule
\makecell[l]{\textit{BS}$(r^*)$} & $\int_{r^*}^{\bar v} (1-G(x))G^{M-1}(x) \, dx$ & \makecell[l]{$-(1-u^*)(u^*)^{M-1}\boldred{v(u^*)}- $ \\ $\int_{u^*}^1 ((M-1)z^{M-2} - M z^{M-1}) \boldred{v(z)}\, dz $}\\
\midrule
\makecell[l]{\textit{RE}$(r^*)$} & \makecell[l]{$M (1-G(r^*))G^{M-1}(r^*)r^{\ast} +$ \\  $\int_{r^*}^{\bar v} x d G^M(x)$ +\\  $M\int_{r^*}^{\bar v} x d \left( (1-G(z))G^{M-1}(z) \right)$} & \makecell[l]{$M (1-u^*)(u^*)^{M-1}\boldred{v(u^*)} +$ \\ $\int_{u^{\ast}}^1 M z^{M-1} \boldred{v(z)}\, dz$ + \\  $M \int_{u^{\ast}}^1 \left( ((M-1)z^{M-2} - mz^{M-1}) \right) \boldred{v(z)}\, dz$}\\
\midrule
\makecell[l]{$\beta(v|r^*)$} & $ \frac{G^{M-1}(r^*)r^*}{G^{M-1}(v)} + \int_{r^*}^v \frac{x}{G^{M-1}(v)} d G^{M-1}(x)$ & $ \frac{(u^*)^{M-1}}{u^{M-1}} \boldred{v(u^*)} + \int_{u^*}^u \frac{(M-1)z^{M-2}}{u^{M-1}} \boldred{v(z)} \,dz $\\
\midrule
\makecell[l]{$\pi(v|r^*)$} & $ \int_{r^*}^v G^{M-1}(x) dx$ & \makecell[l]{$ u^{M-1} \boldred{v(u)} - (u^*)^{M-1} \boldred{v(u^*)} -$\\ $(M-1)\int_{u^{\ast}}^{u} z^{M-2} \boldred{v(z)} dz $}\\
\bottomrule
\end{tabularx}
\caption{Typical counterfactuals in the baseline model.}
\label{table_one}
\end{table}

It can be seen that all the aforementioned counterfactuals are complicated, nonlinear functionals of the primitives. However, using change of variables $z=G(x)$ (i.e., passing to the \emph{ranks} of valuations from their levels) and denoting $u^*=G(r^*)$ yields expressions that are \emph{linear} in the quantile function $v(\cdot)$, see \cref{table_one}. 

Interestingly, since the bid $\beta(v|r^*)$  and the bidder's interim surplus $\pi(v|r^*)$ are monotone in $v$, their median values can be easily computed as
$$v\left(\frac{1+u^*}{2}\right) - \int_{u^*}^{\frac{1+u^*}{2}} \frac{z^{M-1}}{(\frac{1+u^*}{2})^{M-1}} \ dv(z), \quad \int_{u^*}^{\frac{1+u^*}{2}} z^{M-1} d v(z),$$
which are also linear in $v(u)$. This makes $v(u)$ a key object for the counterfactual analysis.

\subsection{Random number of bidders and binding reserve prices}

Let there be $M \geqslant 2$ \textit{potential bidders}, and $m \leqslant M$ \textit{active bidders} in the auction. A potential bidder becomes active if she passes an exogenous and anonymous selection procedure, such as, for example, an (existing) binding reserve price $\underline r$. We can interpret the active bidders as the ones observed by the econometrician and $\underline r$ as the lower end of the support of bids in the data. Denote the resulting probability of observing $m$ active bidders by $p_m$ so that the expected number of active bidders equals $\tilde M := \sum_{m=1}^M m p_m$. Crucially, bidders do not observe how many opponents they face. We are interested in equilibrium behavior when switching to a (counterfactual) higher reserve price $r^* > \underline r$.\footnote{If the counterfactual reserve price exceeds the valuation of an active bidder, she is still considered active.} 

A subtle difficulty with this approach is that while the likelihood of observing $m$ active bidders, as perceived by the econometrician, is equal to $p_m$, the same likelihood, as perceived by the participant (i.e., conditional on being active), is equal to $\tilde p_{m} := m p_m / \tilde M$. We will refer to $p_m$ and $\tilde p_m$ as the \textit{objective} and \textit{subjective} probabilities. To stay within the rational behavior framework, the bidder's beliefs have to be aligned to the subjective probability.\footnote{Although the equilibrium beliefs $(\tilde p_m)_{m=1}^{M}$ depend on the reserve price $\underline r$ as well as the unspecified selection procedure, its nature is irrelevant as long as the beliefs are identical and do not depend on the identity of the bidder nor his value, see, e.g., \cite{krishna2009auction}.} This would guarantee coherent formulas for the bidders's surplus and the auctioneer's revenue. 

We impose a modified version of \cref{ass:distr-bids}, with the only difference being that the relevant support of $G$ is $[\underline r, \bar v]$ rather than $[0, \bar v]$ and the beliefs of the bidders satisfy $\tilde p_0 + \tilde p_1 \neq 1$. The primitives $\underline r, G,(p_m)_{m=0}^{M}$ of the model are common knowledge. 

The ex-ante total surplus, using objective probabilities, is
$\textit{TS}(r^*) \bydef \int_{r^*}^{\bar v} {x} d \left( A_2(G(x)) \right)$, where $A_2(u) \bydef \sum_{m=1}^{M} p_m u^m$.
To derive the interim bidder's surplus, conditional on being active, using subjective probabilities, denote $A_1(u) \bydef \sum_{m=1}^M \tilde p_{m} u^{m-1}$ and write $\pi(v|r^*) \bydef \int_{r^{\ast}}^{v} A_1(G(x)) \, d x$.\footnote{Note that, with correctly specified beliefs, the coefficients of $A_1(u)$ are proportional to that of $A'_2(u)$.} Taking expectation over the distribution of $v$ and integrating by parts, the ex-ante bidder's surplus, conditional on being active, is $\textit{BS}(r^*) \bydef  \int_{r^{\ast}}^{\bar v}A_3(G(x)) dx$, where $A_3(u) \bydef (1-u) A_1(u)$.

The seller's expected revenue under the counterfactual reserve price is equal to the difference between the total surplus and the bidder's surplus (conditional on being active) times the expected number of active bidders:
\begin{align*}
&\textit{RE}(r^*) \bydef \textit{TS}(r^*) - \tilde M \cdot \textit{BS}(r^*) = \int_{r^{\ast}}^{\bar v} {x} d \left(A_2(G(x)) \right) - \tilde M \int_{r^{\ast}}^{\bar v} A_3(G(x)) d x = \\
&= \tilde M A_3(G(r^{\ast}))r^{\ast} + \int_{r^{\ast}}^{\bar v} {x} d \left( A_2(G(x))+ \tilde M A_3(G(x)))\right).
\end{align*}
The exact same formula can be obtained via revenue equivalence with the second price auction or via expectation of the strategy over the highest value of the active bidder.


The equilibrium bidding strategy $\beta(v|r^*)$ with reserve price $r^*$ can be derived, using subjective probabilities, via the envelope conditions:
\begin{gather*}
\beta(v|r^*) = v - \frac{\int_{r^*}^v A_1( G(x))\,dx}{A_1(G(v))} = \frac{A_1( G(r^*))r^* + \int_{r^*}^v x\,dA_1( G(x))}{A_1(G(v))},
\end{gather*}
for all $v \geqslant r^* \geqslant \underline r$. 
This strategy is strictly increasing and twice continuously differentiable. Moreover, it has a strictly positive derivative on $[r^*,\bar v].$\footnote{$\beta'(v|r^*) = \frac{(v-\beta(v|r^*))g(v)}{A(G(v))}$ for all $v > r^*$ and $\beta'(r^*|r^*) = \frac{1}{1 + (A(G)/g)'(r^*)} > 0$, where $A(u) \bydef A_1(u)/A_1'(u)$.} Similar to the baseline model, the inverse bidding strategy can be written either directly or in the quantile form:
\begin{eqnarray}
    v & = & \beta^{-1}(b) = b + \frac{A(F(b))}{f(b)},
    \label{iden1random}\\
    v(u) & = & \beta^{-1}(Q(u)) = Q(u) + A(u) q(u),
    \label{iden2}
\end{eqnarray}
and \cref{prop:distr-of-bids} also holds under the modified version of \cref{ass:distr-bids}.

\begin{table}[t!]
\begin{tabularx}{\textwidth}{l|X|X}
\toprule
\toprule
 & Classical (nonlinear) form & Quantile (linear in $v$) form\\
\midrule
\makecell[l]{\textit{TS}$(r^*)$} & $\int_{r^*}^{\bar v} x \, d \left(A_2(G(x)) \right)$ & $\int_{u^{\ast}}^{1} A_2'(z)\boldred{v(z)} \, dz$\\
\midrule
\makecell[l]{\textit{BS}$(r^*)$} & $ \int_{r^*}^{\bar v} A_3(G(x)) \, dx$ & $- A_3(u^*)\boldred{v(u^*)} - \int_{u^*}^1 A_3'(z) \boldred{v(z)}\, dz $\\
\midrule
\makecell[l]{\textit{RE}$(r^*)$} & \makecell[l]{$\tilde M A_3(G(r^*))r^* +$ \\  $\int_{r^*}^{\bar v} x d \left( A_2(G(x)) + \tilde M A_3(G(x))\right)$} & \makecell[l]{$\tilde M A_3 (u^{\ast})\boldred{v(u^*)} +$ \\ $\int_{u^{\ast}}^1 \left( A_2'(z) + \tilde M A_3'(z) \right) \boldred{v(z)}\, dz$}\\
\midrule
\makecell[l]{$\beta(v|r^*)$} & $ \frac{A_1(G(r^*))r^*}{A_1(G(v))} + \int_{r^*}^v \frac{x}{A_1(G(v))} d A_1(G(x))$ & $ \frac{A_1(u^*)}{A_1(u)} \boldred{v(u^*)} + \int_{u^*}^u \frac{A_1'(z)}{A_1(u)} \boldred{v(z)} \,dz $\\
\midrule
\makecell[l]{$\pi(v|r^*)$} & $ \int_{r^{\ast}}^{v} A_1(G(x)) d x$ & \makecell[l]{$ A_1(u) \boldred{v(u)} - A_1(u^*) \boldred{v(u^*)} -$\\ $\int_{u^{\ast}}^{u} A'_1(z) \boldred{v(z)} dz $}\\
\midrule
\makecell[l]{$p_m(r^*)$} & $\sum_{i = m}^{M} \binom{i}{m} (1-G(r^{*}))^m G(r^{*})^{i-m}p_{i} $ & $\sum_{i = m}^{M}  \binom{i}{m} (1-u^{*})^m(u^{*})^{i-m}p_{i}$\\
\bottomrule
\end{tabularx}
\caption{Typical counterfactuals when the number of bidders is random.}
\label{table_one_random}
\end{table}
Finally, the counterfactual participation pattern using objective probabilities is
$$p_m(r^*) := \sum_{i = m}^{M} \binom{i}{m} (1-G(r^{*}))^m G(r^{*})^{i-m}p_{i}.$$

Using the change of variables $z=G(x)$ and denoting $u^*=G(r^*)$ yields expressions that are \emph{linear} in the quantile function $v(\cdot)$, see \cref{table_one_random}. 

\subsection{Risk aversion and opportunity cost}

Suppose the number of bidders entering an auction is random.
Suppose also that the bidders have constant relative risk aversion (CRRA) utility with risk-aversion parameter $\eta$, and let the auctioneer have the opportunity cost $c$.
The equilibrium bidding strategy $\beta(\cdot|r^*)$ of the active bidder with reserve price $r^*$ can be characterized 
by applying the Envelope Theorem to the maximization of  $(v - b)A^{\frac{1}{1-\eta}}_1(G(\beta^{-1}(b)))$:
\begin{gather*}\beta(v|r^*) = v - \frac{\int_{r^*}^v A^{\frac{1}{1-\eta}}_1( G(x))\,dx}{A^{\frac{1}{1-\eta}}_1(G(v))} = \frac{A^{\frac{1}{1-\eta}}_1(G(r^*))r^* + \int_{r^*}^v x dA^{\frac{1}{1-\eta}}_1( G(x))}{A^{\frac{1}{1-\eta}}_1(G(v))},
\end{gather*}
for all $v \geqslant r^* > \underline r$. Clearly, the strategy is still linear in $v(\cdot)$.

Due to risk aversion, we cannot derive revenue as in the previous sections, nor can we use the Revenue Equivalence between the first-price and the second-price auction, see \cite{krishna2009auction}.
Instead, similar to \cite{zincenko2024estimation}, we can derive revenue as the expectation of the highest bid net the opportunity cost over the distribution of the highest active bidder's value, $\textit{RE}(r^*) = \int^{\overline v}_{r^*} (\beta(x) - c) d A_2(G(x))$.
Then,
\begin{equation}
\textit{RE}(r^*) = \int^{\overline v}_{r^*} (x-c) d A_2(G(x)) - \int^{\overline v}_{r^*} \int_{r^*}^v A^{\frac{1}{1-\eta}}_1( G(x))\,dx \, dA_4(G(v)),
\end{equation}
where $A_4(x) \bydef \int_0^{x} A^{\frac{-1}{1-\eta}}_1(x)d A_2(x) = \tilde M \int_0^{x} A^{\frac{-\eta}{1-\eta}}_1(x) dx$ because $A'_2(x) = \tilde M A_1(x)$. Finally, using integration by parts similar to the risk-neutral case,
\begin{gather}
\int^{\overline v}_{r^*} \int_{r^*}^v A^{\frac{1}{1-\eta}}_1( G(x))\,dx dA_4(G(v)) 
= - \int^{\overline v}_{r^*} A^{\frac{1}{1-\eta}}_1( G(x)) (A_4(G(x))-A_4(1)) dv= \\
= A^{\frac{1}{1-\eta}}_1( G(r^*)) (A_4(G(r^*))-A_4(1))r^* + \int^{\overline v}_{r^*} x d A^{\frac{1}{1-\eta}}_1( G(x)) A_4(G(x)).
\end{gather}
Therefore, with risk aversion, expected revenue is still linear in $v(\cdot)$.

Finally, although bidder's interim surplus $\pi(v|r^*) = (\int_{r^*}^v A^{\frac{1}{1-\eta}}_1(G(x)) dx)^{1-\eta}$ is not linear in $v(\cdot)$, this problem can be circumvented since tests about the maximum of $\pi(v|r^*)$ over $r^*$ can be constructed from tests about the maximum of $(\pi(v|r^*))^{\frac{1}{1-\eta}}$.

\subsection{Data generating process}\label{sec:dgp}

The observed data is a random sample of bids $\{b_{il},\,\, i=1,\dots,m_l,\,\, l=1,\dots,L \}$, where $b_{il}$ denotes the bid submitted by the $i$-th participant in the $l$-th auction. All the auctions are ex-ante symmetric and independent, and $m_l$ is the number of participants in the $l$-th auction. For brevity, we denote the (random) sample size by $n=n(L) = \sum_{l=1}^L m_l$ and define
\begin{align}
    b_1\bydef b_{11}, \,\, b_2\bydef b_{21}, \ \dots \ ,b_n\bydef b_{m_L L}.
\end{align}
Note that since the bidders do not know the realizations of the number of active bidders, the samples $\{b_1,\dots,b_n\}$ and $\{m_1,\dots,m_L\}$ are independent. Besides, as $L \to \infty$, the sample size $n(L) \to \infty$ with probability one. Therefore, without loss of generality, we condition our subsequent exposition on a realization $\{m_l\}_{l=1}^\infty$ such that $n(L) \to \infty$. This has the following important implication: although the auxiliary functions $A_1,A_2,A_3,A$ and the constant $a$ need to be estimated from the data, we can assume that they are known since their estimators only depend on the conditioning variables $m_1,\dots,m_L$, see equations \eqref{eq:A-1-estimator}-\eqref{eq:phat} below.

%% file: valuation_quantiles.tex
\section{Estimation and inference for value quantiles}
\label{sec:BK}

As explained in \cref{sec:counterfactuals}, the value quantile function $v(\cdot)$ is the key object needed for the counterfactual analysis. In this section we develop the asymptotic theory for its natural (plug-in) estimator. To define the estimator, we need to introduce two auxiliary objects.

The first object is the kernel estimator of the bid quantile density $q(u)$, defined by
\begin{align}
\hat q_h(u) &\bydef \int_{0}^{1} K_h(u-z) \, d\hat Q(z), \quad u\in [0,1].
\end{align}
Here $K$ is a compactly supported kernel, $K_h(z) \bydef h^{-1}K\left(h^{-1}z\right)$, $h>0$ is a bandwidth, and $\hat Q(u)$ is the empirical bid quantile function,
\begin{align}\label{eq:Q-est}
    \hat Q(u) =
    \begin{cases}
    b_{(\left\lfloor{nu}\right\rfloor +1)}, \quad u\in [0,1),\\
    b_{(n)}, \quad u=1,
    \end{cases}
\end{align}
where $b_{(1)}\le\dots\le b_{(n)}$ are the order statistics of the observed bids $b_1,\dots,b_n.$ We note that $\hat q_h$ takes the form of a weighted sum of \emph{bid spacings} $b_{(i+1)}-b_{(i)}$,
\begin{equation}\label{eq:q-est}
\hat q_h(u) = \sum_{i=1}^{n-1} K_h(u-i/n) \bigPar{b_{(i+1)}-b_{(i)}}.
\end{equation}
This estimator was previously studied by \citet{siddiqui1960distribution} and \citet{bloch1968simple} for the case of rectangular kernel, and by \citet{falk1986estimation}, \citet{welsh1988asymptotically}, \citet{csorgHo1991estimating}, and \cite{jones1992estimating} for general kernels.

The second auxiliary object is the plug-in estimators of $A_1,A_2,A_3,A$, and $\tilde M$ defined by
\begin{align}
    \check A_1(u) &\bydef \check A'_2(u)/\check M = \sum_{m=1}^{M} \frac{m \check p_m}{\check M} u^{m-1},\quad \check A_2(u) \bydef \sum_{m=1}^{M} \check p_m u^m,\label{eq:A-1-estimator}\\
    \check A_3(u) & \bydef (1-u) \check A_1(u), \quad \check A(u) \bydef  \frac{\check A_1(u)}{\check A_1'(u)}, \quad \check M \bydef \sum_{m=1}^M m \check p_m,
\end{align}
where $\check p_m$ is the empirical frequency of auctions with $m$ bidders,
\begin{align}
    \check p_m \bydef \frac{1}{L}\sum_{l=1}^L 1(m_l=m), \quad m=1,\dots,M. \label{eq:phat}
\end{align}
We use the ``check'' (as opposed to ``hat'') notation here to highlight that $\check A_1, \check A_2, \check A_3, \check A, \check M$ are treated as known since, as explained in \cref{sec:dgp}, our analysis is conditional on $m_1,\dots,m_L$.

Given $\check A$ and $\hat q_h$, we define our estimator of the value quantile $v(u)$ by
\begin{align}
    \hat v_h(u) &\bydef \hat Q(u) + \check A(u) \hat q_h(u), \quad u\in[0,1].\label{eq:def-v-hat}
\end{align}

We note that $\hat v_h$ consists of two parts: (i) the empirical quantile function $\hat Q$ that is uniformly consistent and converges to a Gaussian process in $\linfzeroone$ at the parametric rate $n^{-1/2}$, and (ii) the kernel quantile density $\hat q_h$ that is uniformly consistent only away from the boundary $\{0,1\}$ and does \emph{not} converge to a (tight) limit in $\ell^\infty[\varepsilon,1-\varepsilon]$ even if $\varepsilon>0$, but converges pointwise to a Gaussian limit at the nonparametric rate $(nh)^{-1/2}$, see the proof of \cref{Thm:BK-expansion}.
Therefore, the first-order asymptotic properties of $\hat v_h$ are determined by the kernel quantile density $\hat q_h$.

We impose the following assumptions.


\begin{assumption}[Kernel function] \label{ass:kernel}
\text{ }
\begin{enumerate}
    \item \label{ass:kernel-basic}$K: \R \to \R$ is a nonnegative function such that
    \begin{align}
        \int_\R K(z)\, dz=1 \text{ and } R_K\bydef \int_{\R} K(z)^2 \, dz < \infty.
    \end{align}
    \item \label{ass:kernel-Lipschitz}$K$ is a Lipschitz function supported on the interval $[-1,1]$. 
\end{enumerate}
\end{assumption}

\begin{assumption}[Bandwidth, estimation] \label{ass:band-large} The bandwidth $h=h_n$ is such that $h\to 0$ and there exist $c>0$ and $\alpha>0$ such that $h_n \ge c n^{-1/2+\alpha}$ for all $n$.
\end{assumption}

\begin{assumption}[Bandwidth, inference] \label{ass:band-small} The bandwidth $h=h_n$ is such that there exist $C>0$ and $\beta>0$ such that $h_n \le C n^{-1/3-\beta}$ for all $n$.
\end{assumption}

\cref{ass:kernel}.\ref{ass:kernel-basic} states that $K$ is a valid, square-integrable PDF.
\cref{ass:kernel}.\ref{ass:kernel-Lipschitz} is standard in the literature on strong approximations of local empirical processes \citep[see, e.g.,][]{rio1994local}.
In particular, it implies that $K$ is a function of bounded variation, which is crucial in our derivation of the BK expansion.
\Cref{ass:band-large} is sufficient if the goal is estimation, because it leads to consistency of $\hat v_h$ and allows for MSE-optimal estimation of $\hat v$, where the optimal bandwidth is $h=O(n^{-1/5})$, see Theorem 2.2 of \citet{csorgHo1991estimating}.
On the other hand, if the goal is inference, then our approach is to impose \cref{ass:band-small} (undersmoothing) to eliminate the bias in $\hat v_h$ and related nonsmooth counterfactuals.



\subsection{The Bahadur--Kiefer expansion}\label{sec:BK-for-v}

In this section, we derive the Bahadur--Kiefer (i.e. almost sure, uniform, asymptotically linear) representation of the form
\begin{align}\label{E:BK-v-abstract}
    \frac{\sqrt{nh}(\hat v_h(u)-v(u))}{\hat q_h(u)} = -\check A(u) \eG_{n,h}(u) + R_n(u), \quad u\in[h,1-h],
\end{align}
where
\begin{align}
    \eG_{n,h}(u) \bydef \sqrt{nh} \cdot \frac{1}{n} \sumin \left[K_h\left(u-F(b_i)\right) - \E K_h\left(u-F(b_i)\right) \right] \label{eq:Gnh-def}
\end{align}
and the remainder $R_n(u)$ converges to zero a.s. uniformly in $u\in[h,1-h]$ with an explicit rate.

The key feature of this representation is that the main term is fully \emph{known} and \emph{pivotal}: its distribution does not depend on the data generating process since $U_i \bydef F(b_i) \sim \text{iid Uniform}[0,1]$. Heuristically, this suggests that the distribution of the left-hand side under \emph{any} DGP is a valid approximation for its true distribution. Indeed, in \cref{sec:inference}, we show the validity of such approximation by combining pivotality with the anti-concentration theory of \citet{chernozhukov2014gaussian}. This leads to a simple algorithm for the confidence bands on the quantile function $v(\cdot).$


To derive this representation, we rely on the classical BK expansion of the quantile function \citep{bahadur1966note,kiefer1967bahadur},
\begin{align}\label{E:BK-general}
 \hat Q(u)-Q(u) &= - q(u)\left( \hat F(Q(u)) - u \right) + r_n(u),\\
 \text{where } r_n(u) &= O_{a.s.}\left( n^{-3/4} \ell(n) \right) \text{ uniformly in } u \in [0,1].
\end{align}
Here $\ell(n)=(\log n)^{1/2}(\log \log n)^{1/4}$ is a logarithmic offset factor that arises due to the uniform nature of the approximation and may often be disregarded in practice. Note that the BK expansion represents a \emph{nonlinear} estimator $\hat Q(u)$ as a sum of the \emph{linear} estimator --- the empirical distribution function $\hat F(Q(u))$ --- and the remainder $r_n(u)$ that converges to zero a.s. uniformly at a nonparametric (slow) rate $n^{-3/4}\ell(n)$.

\begin{thm}[Bahadur-Kiefer expansion for value quantiles] \label{Thm:BK-expansion}
Under the Assumptions \ref{ass:distr-bids} and \ref{ass:kernel}, the estimator $\hat v_h(u)$ has the representation
\begin{align}
    Z_n(u) = Z_n^*(u) + R_n(u), \quad u\in[h,1-h], \label{E:BK}
\end{align}
where
\begin{align}
    Z_n(u) &\bydef \frac{\sqrt{nh}\left(\hat v_h(u)-v(u)\right)}{\hat q_h(u)}, \quad
    Z_n^*(u) \bydef -\check A(u) \eG_{n,h}(u), \\
    R_n(u) &= O_{a.s.}\left(n^{1/2}h^{3/2} + h^{1/2} + h^{-1/2}n^{-1/4}\ell(n)\right) \text{ uniformly in } u\in[h,1-h]. \label{E:BK-remainder}
\end{align}
\end{thm}

\begin{rem}[BK expansion for quantile density]
The proof of the preceding theorem also implies the BK expansion for the normalized quantile density $\sqrt{nh}\left(\hat q_h(u)-q_h(u)\right)$, which may be of independent interest. In this case, the right-hand side does not have the factor $\check A(u)$, while the term $h^{1/2}$ in the remainder rate can be replaced by the faster term $h\log h$.
\end{rem}


We note that two types of biases arise in the estimation of $v(\cdot)$.
The first type of bias is the \emph{smoothing bias} $\E \hat v_h(u)-v(u)$ which manifests in the term $n^{1/2}h^{3/2}$ in the remainder rate. This bias can be eliminated by undersmoothing $h=o(n^{-1/3})$, i.e. choosing a (suboptimally) small bandwidth such that $\sqrt{nh}(\E \hat v_h(u)-v(u)) \to 0$, which is our \Cref{ass:band-small}.

The other type of bias is the \emph{boundary bias}, stemming from the estimator $\hat v_h(u)$ being inconsistent when $u$ is close to the boundary $\{0,1\}$ of its domain $[0,1]$. Because our interest is in valid hypothesis testing, and not the confidence bands \emph{per se}, we can eliminate this bias by introducing the trimming $u \in [h,1-h]$ while maintaining the validity of inference procedures based on the representation \eqref{E:BK}.


\subsection{Inference on value quantiles}
\label{sec:inference}

\cref{Thm:BK-expansion} allows us to construct pointwise confidence intervals and uniform confidence bands for the value quantile function. In particular, the following corollary provides the asymptotic distribution of the estimator of a fixed valuation quantile.

\begin{cor}\label{Cor:pointwise-asy}
Under the Assumptions \ref{ass:distr-bids}, \ref{ass:kernel}, \ref{ass:band-large}, and \ref{ass:band-small}, we have, for every $u\in(0,1)$,
\begin{align}
    &\sqrt{nh}\left( \hat v_h(u)-v(u) \right) \weakto N(0, V(u)),\\
    &V(u) \bydef A^2(u)q^2(u)R_K.
\end{align}
\end{cor}
Special cases of this result for the quantile density estimator $\hat q_h$ were derived by \citet{siddiqui1960distribution} and \citet{bloch1968simple}. It implies that a confidence interval of nominal confidence level $1-\alpha$ for $v_h(u)$ can be constructed as
\begin{align}
    \left[ \hat v_h(u) - \frac{\check  A(u)\hat q_h(u) \sqrt{R_K}}{\sqrt{nh}} z_{1-\frac{\alpha}{2}}, \quad \hat v_h(u) + \frac{\check A(u) \hat q_h(u) \sqrt{ R_K}}{\sqrt{nh}} z_{1-\frac{\alpha}{2}} \right],
\end{align}
where $z_{1-\frac{\alpha}{2}}$ is the standard normal quantile of level $1-\frac{\alpha}{2}$.


We now turn to the problem of uniform inference on $v(\cdot)$.

Note that if the process $Z_n$ converged weakly in $\linfh$ to a known (or estimable) process $Z$, this would have enabled the construction of asymptotically valid confidence bands by using the quantiles of $\sup_u | Z(u) |$ as critical values.\footnote{For one-sided confidence bands, one would use the quantiles of $\sup_u Z_n(u)$ instead.} Unfortunately, although $Z_n(u)$ is asymptotically Gaussian at each point $u \in (0,1)$, it does \emph{not} converge in $\linfh$. This follows from the fact that the main term $Z_n^*$ in the BK expansion is the scaled kernel density process, which is known to lack functional convergence \citep[see, e.g.,][]{rio1994local}.\footnote{For an example of a sequence of stochastic processes on $[0,1]$ that weakly converges pointwise, but not in $\linfzeroone$, consider $X_n(u)=h_n^{-1/2}\left(B\bigPar{u+h_n}-B\bigPar{u}\right)$, where $B$ is the Brownian motion, $h_n \to 0$ and $u\in(0,1)$. Clearly, $X_n(u) \weakto N(0,1)$ for all $u$, but, by L\'{e}vy's modulus of continuity theorem, $\sup_u \left|X_n(u)\right| \to \infty$ a.s., and so there is no convergence in $\linfzeroone$.}  In such a case, there are two common ways to circumvent the problem and derive valid confidence intervals.

One approach is to derive the asymptotic distribution of a normalized version of $\sup_u Z_n(u)$ using extreme value theory, and then rely on the knowledge of the normalizing constants to construct the confidence band. In the case of kernel and histogram density estimation, this approach was pioneered by \cite{smirnov1950construction} and \cite{bickel1973some}.\footnote{For a nonasymptotic version of Smirnov-Bickel-Rosenblatt extreme value theorem, see \citet[][Theorem 1.2]{rio1994local}.}
However, convergence to the asymptotic distribution turns out to be very slow, leading to the coverage error of the resulting confidence band to be $O(1/\log n)$), as shown by \citet{hall1991convergence}.

The other approach is to rely on finite-sample approximations for (the distribution of) the supremum
\begin{align}
 W_n &\bydef \sup_{u\in [h,1-h]} \left|Z_n(u)\right|.   \label{eq:w-n}
\end{align}
If such an approximation admits simulation, it can be used for the construction of confidence bands. This is the approach we take in this paper.

We consider two types of approximations, both of which are \emph{pivotal}, and hence allow for simulation. One is simply the supremum of the linear term $Z_n^*$, viz.
 \begin{align}
     W_n^* &\bydef \sup_{u\in [h,1-h]} \left|Z_n^*(u)\right|. \label{eq:w-n-star}
 \end{align}
The other is the supremum of $Z_n$ under an alternative, uniform[0,1] distribution of bids, viz.
\begin{align}
    W_n^{U[0,1]} &\bydef \sup_{u\in [h,1-h]} \left|Z_n^{U[0,1]}(u)\right|, \label{eq:w-n-u}
\end{align}
where $Z_n^{U[0,1]}(u)$ is the process $Z_n(u)$ calculated using the pseudo-sample 
\begin{align}
    \{\tilde b_i\}_{i=1}^n \sim \text{iid Uniform}[0,1]. \label{eq:unif-pseudo-bids}
\end{align}
This approximation is nonstandard and makes use of the asymptotic pivotality of $W_n$. In principle, any distribution of the pseudo-bids rationalized by a value distribution satisfying Assumption \ref{ass:distr-bids} can be chosen; however, the uniform distribution is convenient since, in this case, we have, for all $u\in [0,1]$,
\begin{align}
    Q(u) \bydef u, \quad q(u) \bydef 1, \quad v(u) \bydef u + A(u),
\end{align}
and hence
\begin{align}
    Z_n^{U[0,1]}(u) \bydef \sqrt{nh}\left( \hat v_h(u; \{\tilde b_i\}_{i=1}^n)-u - \check A(u)\right).
\end{align}

We emphasize that it is not immediate that the distributions of $W_n^*$ and $W_n^{U[0,1]}$ approximate the distribution of $W_n$ in a way that guarantees the validity of the associated confidence bands
\begin{align}
    &\left[ \hat v_h(u) - \frac{\hat q_h(u)c_{n,1-\alpha/2}}{\sqrt{nh}}, \,\, \hat v_h(u) + \frac{\hat q_h(u) c_{n,1-\alpha/2}}{\sqrt{nh}} \right], \quad u\in [h,1-h], \label{E:conf-band}
\end{align}
where $c_{n,1-\alpha/2}$ is the $(1-\alpha/2)$-quantile of either $W_n^*$ or $W_n^{U[0,1]}$. Indeed, note that \cref{Thm:BK-expansion} implies the inequality
\begin{align}
   \left| W_n - W_n^* \right| = \left| \sup_{u \in [h,1-h]} |Z_n(u)| -\sup_{u \in [h,1-h]} |Z_n^*(u)| \right| \le \sup_{u \in [h,1-h]} |R_n(u)|,
\end{align}
and hence the coupling
\begin{align}\label{E:coupling-abstract}
    W_n = W_n^* + r_n,
\end{align}
where $r_n$ tends to zero a.s. at a known rate. If one could show that this implies Kolmogorov convergence
\begin{align}\label{E:Kolm-convce-abstract}
    \sup_{t\in \R} \left|\Pb(W_n\le t)-\Pb(W_n^*\le t) \right| \to 0,
\end{align}
then the confidence bands based on $W_n^*$ would be valid. However, \eqref{E:Kolm-convce-abstract} need not follow from \eqref{E:coupling-abstract} even if the a.s. convergence rate of $r_n$ is very fast, unless further conditions are imposed on $W_n^*$.

As an illustration of this phenomenon, consider an abstract example $W_n=n^{-1} U$, $W_n^*=n^{-1} (U-1)$, where $U\sim \text{Uniform}[0,1]$. Then $r_n=W_n - W_n^*= n^{-1}$, but
\begin{align*}
    \Pb(W_n\le 0) - \Pb(W_n^*\le 0) = 0-1 = -1 \not\to 0,
\end{align*}
and so \eqref{E:Kolm-convce-abstract} does not hold.
On the other hand, if $W_n^*$ had an absolutely continuous asymptotic distribution $\mathcal{D}$, then the CDF of $W_n$ would converge to the CDF of $\mathcal{D}$ pointwise, and hence the quantiles of $\mathcal{D}$ would serve as valid critical values.

Therefore, intuitively, a certain degree of \emph{anti-concentration} of $W_n^*$ is needed to guarantee that the coupling \eqref{E:coupling-abstract} implies Kolmogorov convergence \eqref{E:Kolm-convce-abstract} and hence validity of simulated critical values. The anti-concentration literature mainly focuses on Gaussian processes, while the process $Z_n^*$ is non-Gaussian. Fortunately, $Z_n^*$ is the normalized kernel density estimator for uniform data, which is a well-studied process. In particular, we rely on the seminal work \citet{chernozhukov2014gaussian} to establish a coupling of $W_n$ with the supremum of a Gaussian process and show that the latter exhibits sufficient anti-concentration. We then argue that an identical argument works for $W_n^*$. Finally, the pivotality of $W_n^*$ and the coupling \eqref{E:coupling-abstract} imply the Kolmogorov convergence for $W_n^{U[0,1]}$. Formally, we have the following result. 

\begin{thm}\label{Thm:uni-conf-bands}
Under the Assumptions \ref{ass:distr-bids}, \ref{ass:kernel}, \ref{ass:band-large}, and \ref{ass:band-small},
\begin{align}
    &\sup_{x \in \R} \left| \Pb(W_n\le x) - \Pb(W_n^*\le x) \right| \to 0,\\
    &\sup_{x \in \R} \left| \Pb(W_n\le x) - \Pb(W_n^{U[0,1]}\le x) \right| \to 0,
\end{align}
and hence the confidence bands \eqref{E:conf-band} are asymptotically valid and exact.
\end{thm}



\begin{rem}
To construct the \emph{one-sided} confidence bands, we note that the same result holds with $W_n$, $W_n^*$, and $W_n^{U[0,1]}$ replaced by $\sup_{u\in[h,1-h]} Z_n(u)$, $\sup_{u\in[h,1-h]} Z_n^*(u)$, and \newline
$\sup_{u\in[h,1-h]} Z_n^{U[0,1]}(u)$, respectively.
\end{rem}

%% file: functionals.tex
\section{Estimation and inference for counterfactuals}
\label{sec:func}

In this section, we develop the asymptotic theory for the counterfactuals in \cref{table_one} which heavily relies on the analysis of the estimator of value quantiles in the previous section.

Clearly, every such counterfactual has the general form
\begin{align}
    T(u^*) \bydef \varphi(u^*)v(u^*) + \int_{u^*}^1 \psi(x)v(x) \, dx, \label{eq:T-popul}
\end{align}
where $\varphi$ and $\psi$ are continuously differentiable functions on $[0,1]$ that only depend on the auxiliary objects $A_1,A_2,A_3,A$ (or their derivatives). As an example, for the total expected surplus $\varphi(x) \equiv 0$ and $\psi(x) = A_2'(x)$, while for the expected revenue $\varphi(u^*)=\tilde M A_3(u^*)$ and $\psi(x) = A_2'(x)+\tilde M A_3'(x)$.

The representation \eqref{eq:T-popul} implies $T(u^*)$ is a (weighted) sum of two continuous linear functionals of $v$ of different smoothness: (i) evaluation at a point $v(u^*)$ and (ii) integration
\begin{align}
    S_{\psi}(u^*) \bydef \int_{u^*}^1 \psi(x)v(x) \, dx. \label{eq:S-type-func}
\end{align}
The natural estimators of the two components have fundamentally different asymptotic properties. Namely, in \cref{sec:BK-for-v} we showed that the less smooth functional $v(u^*)$ is only estimable at the nonparametric rate $(nh)^{-1/2}$ and does \emph{not} converge in $\ell^\infty[\varepsilon,1-\varepsilon]$ even for $\varepsilon>0$. On the other hand, we will show in \cref{subsec:S-counterf} that the smoother functional $S_{\psi}(u^*)$ is estimable at the parametric rate $n^{-1/2}$ and converges to a Gaussian process in $\linfzeroone$. We will combine the two results in \cref{subsec:T-counterf} to show that, whenever $\varphi \neq 0$, inference on $T$ can be performed similarly to that on $v$.


\subsection{Smooth ($S$-type) counterfactuals}\label{subsec:S-counterf}
First, let us consider estimation and inference for functionals \eqref{eq:S-type-func}, where $\psi:\,[0,1]\to \R$ is a known, continuously differentiable function.

To motivate our estimator, use integration by parts to rewrite
\begin{align}
    S_\psi(u^*) &= \int_{u^*}^{1} \psi(u) Q(u)\,du + \int_{u^*}^{1} \check A(u) \psi(u) \,d  Q(u)\\
    &= \int_{u^*}^{1} \psi(u) Q(u)\,du + \check A(u)\psi(u) Q(u)\vert_{u^*}^{1} -  \int_{u^*}^{1} Q(u)\,d[\check A(u)\psi(u)] =\\
    &= \int_{u^*}^{1} \psi(u)Q(u)\,du + \check A(u)\psi(u) Q(u)\vert_{u^*}^{1} -  \int_{u^*}^{1} Q(u)(\check A'(u)\psi(u) + \check A(u)\psi'(u))\,du  \\
    &= \int_{u^*}^{1} \chi_\psi(u) Q(u)\,du - \check A(u^*)\psi(u^*) Q(u^*) + \check A(1)\psi(1) Q(1),
\end{align}
where we denote
\begin{align}
    \chi_\psi(u) \bydef (1-\check A'(u))\psi(u)-\check A(u)\psi'(u). \label{eq:chi-def}
\end{align}

Note that the latter formula expresses $S_\psi(u^*)$ as a continuous linear functional of the quantile function $Q$, which is estimable at a parametric rate. This leads to a natural estimator of $S_\psi$ that does not contain tuning parameters. Namely, for any $u^*\in [0,1]$, we define the estimator
\begin{align}
    \hat S_\psi(u^*) &\bydef \int_{u^*}^{1} \chi_\psi(u) \hat Q(u)\,du - \check A(u^*)\psi(u^*) \hat Q(u^*) + \check A(1)\psi(1) \hat Q(1) \\
    &= \sum_{i = \lfloor nu^*\rfloor}^{n-1} b_{(i+1)} \int_{\frac{\max(i,nu^*)}{n}}^{\frac{i+1}{n}} \chi_\psi(u) \,du  - \check A(u^*)\psi(u^*) b_{(\lfloor nu^*\rfloor + 1)} + \check A(1)\psi(1) b_{(n)}. \label{eq:S-hat}
\end{align}

The following theorem establishes standard Gaussian process asymptotics for $\hat S_\psi$.

\begin{thm}\label{thm:S-unif-asy-distr}
Under the Assumptions \ref{ass:distr-bids}, \ref{ass:kernel}, \ref{ass:band-large}, and \ref{ass:band-small},
\begin{align}
\sqrt{n}(\hat S_\psi(\cdot) - S_\psi(\cdot)) \weakto \G_{\psi,q}(\cdot) \text{ in } \linfzeroone,
\end{align}
where $\G_{\psi,q}(\cdot)$ is a tight, centered Gaussian process on $[0,1]$ with the covariance function
\begin{align}
\E \G_{\psi,q}(u^*)\G_{\psi,q}(v^*) = \text{Cov}\left( f_{u^*}(U),f_{v^*}(U)\right),\quad U\sim \textit{uniform}[0,1], \quad u^*,v^*\in[0,1],
\end{align}
and
\begin{align}
f_{u^*}(U) \bydef - \int_{u^*}^1 \chi_\psi(u)q(u)1(U\le u)\,du + \check A(u^*)\psi(u^*) q(u^*) 1(U\le u^*).
\end{align}
\end{thm}

\begin{rem}
The integral in the expression for $\hat S_\psi(u^*)$ can be replaced by $\chi_\psi(\zeta(i))/n$ for any $\zeta(i) \in \left[ \frac{i}{n}, \frac{i+1}{n} \right]$. This would have no impact on the statement of \cref{thm:S-unif-asy-distr}.
\end{rem}


\subsection{Nonsmooth ($T$-type) counterfactuals}\label{subsec:T-counterf}

Given the asymptotic results for the two components of the generic counterfactual \eqref{eq:T-popul}, we may now turn to estimation and inference on the latter.

To this end, define the estimator
\begin{align}
    \hat T_h(u^*) = \check \varphi(u^*) \hat v_h(u^*) + \hat S_{\check \psi}(u^*), \quad u^*\in[h,1-h],
\end{align}
where $\hat v_h(u^*)$ is defined in \eqref{eq:def-v-hat} and $\hat S_{\check \psi}(u^*)$ is the estimator \eqref{eq:S-hat} with $\psi=\check \psi$. Since $\hat S_{\check \psi}(u^*)$ converges fast, while $\hat v_h(u^*)$ converges slowly, the asymptotics of $\hat T_h(u^*)$ is dominated by the latter, as illustrated by the following theorem.

\begin{thm}\label{thm:T-expansion}
Under the Assumptions \ref{ass:distr-bids} and \ref{ass:kernel}, we have the representation
\begin{align}
    Z_n^T(u^*) = Z_n^{T*}(u^*) + R_n^T(u^*), \quad u^*\in [h,1-h],  \label{eq:T-expansion}
\end{align}
where
\begin{align}
    Z_n^T(u^*) &\bydef \frac{\sqrt{nh}\left( \hat T_h(u^*)-T(u^*) \right)}{ \hat q_h(u^*)}, \quad
    Z_n^{T*}(u^*) \bydef  -\check\varphi(u^*) \check A(u^*) \eG_{n,h}(u^*),\\
    R_n^T(u^*) &\bydef O_{a.s.}\left(n^{1/2}h^{3/2} + h^{1/2} + h^{-1/2} n^{-1/4} \ell(n)\right) \text{ uniformly in } u^*\in[h,1-h].
\end{align}
\end{thm}

\cref{thm:T-expansion} immediately yields the following result on the asymptotic distribution of $\hat T_h(u^*)$ at a fixed point $u^* \in (0,1)$.

\begin{cor}
Under the Assumptions \ref{ass:distr-bids}, \ref{ass:kernel}, \ref{ass:band-large}, and \ref{ass:band-small}, we have, for every $u^*\in(0,1),$
\begin{align}
&\sqrt{nh}(\hat T_h(u^*) - T(u^*)) \weakto N(0,V(u^*)),\\
&V(u^*) = \left(A(u^*) q(u^*) \varphi(u^*) \right)^2 R_K.
\end{align}
\end{cor}

We now consider uniform inference on $T(\cdot)$. Since the estimator $\hat T_h(\cdot)$ does not converge in $\linfh$, but the approximating process $Z_n^{T*}$ is known and pivotal, the methodology will be similar to the case of the valuation quantile function, see \cref{sec:inference}. In particular, we show that valid confidence bands for $T(\cdot)$ can be based on simulation from either (i) the approximating process $Z_n^{T*}$, or (ii) the process $Z_n^T$ under an alternative distribution of bids. To this end, define
\begin{align}
    W_n^T &\bydef \sup_{u^*\in[h,1-h]}\left| Z_n^T(u^*) \right|,\\
    W_n^{T*} &\bydef \sup_{u^*\in[h,1-h]}\left| Z_n^{T*}(u^*) \right|,\\
    W_n^{T,U[0,1]} &\bydef \sup_{u^*\in[h,1-h]} \left|Z_n^{T,U[0,1]}(u^*)\right|,
\end{align}
where $Z_n^{T,U[0,1]}(u^*)$ is the process $Z_n^T(u^*)$ calculated using the pseudo-sample
\begin{align}
    \{\tilde b_i\}_{i=1}^n \sim \text{iid Uniform}[0,1],
\end{align}
cf. equations \eqref{eq:w-n}-\eqref{eq:unif-pseudo-bids}. Define the confidence bands by
\begin{align}
    \left[ \hat T_h(u^*) - \frac{\hat q_h(u^*) c_{n,1-\alpha/2}}{\sqrt{nh}}, \,  \hat T_h(u^*) + \frac{\hat q_h(u^*) c_{n,1-\alpha/2}}{\sqrt{nh}} \right], \quad u\in[h,1-h], \label{eq:T-conf-bands}
\end{align}
where $c_{n,1-\alpha/2}$ is the $(1-\alpha/2)$-quantile of either $W_n^{T*}$ or $W_n^{T,U[0,1]}$.
\begin{thm}\label{Thm:uni-bands-T}
Under the Assumptions \ref{ass:distr-bids}, \ref{ass:kernel}, \ref{ass:band-large}, and \ref{ass:band-small}, we have 
\begin{align}
&\sup_{t\in \R} \left|\Pb(W_n^T \le t) - \Pb(W_n^{T*} \le t) \right| \to 0,\\
&\sup_{t\in \R} \left|\Pb(W_n^T \le t) - \Pb(W_n^{T,U[0,1]} \le t) \right| \to 0,
\end{align}
and hence the confidence bands \eqref{eq:T-conf-bands} are asymptotically valid and exact.
\end{thm}

\begin{rem}
For the purpose of constructing the \emph{one-sided} confidence bands, we note that the same result holds with $W_n^T$, $W_n^{T*}$, and $W_n^{T,U[0,1]}$ replaced by $\sup_{u\in[h,1-h]} Z_n^T(u)$, $\sup_{u\in[h,1-h]} Z_n^{T*}(u)$, and $\sup_{u\in[h,1-h]} Z_n^{T,U[0,1]}(u)$, respectively.
\end{rem}

\begin{rem}[Shape of confidence bands]\label{thm:T-expansion-iota}
Note that, for any function $\iota(\cdot)$ bounded away from zero on $[0,1]$, the representation \eqref{eq:T-expansion} is equivalent to 
\begin{align}
    Z_n^T(u^*)/\iota(u^*) = Z_n^{T*}(u^*)/\iota(u^*) + \tilde R_n^T(u^*), \quad u^*\in [h,1-h],  \label{eq:T-expansion-iota}
\end{align}
where $\tilde R_n^T(u^*)$ has the same uniform convergence rate as $R_n^T(u^*)$.
Similarly to \eqref{eq:T-conf-bands}, the two-sided confidence bands based on such representation are
\begin{align}
    \left[ \hat T_h(u^*) - \frac{\iota(u^*) \hat q_h(u^*) c_{n,1-\alpha/2}}{\sqrt{nh}}, \,  \hat T_h(u^*) + \frac{\iota(u^*) \hat q_h(u^*) c_{n,1-\alpha/2}}{\sqrt{nh}} \right], \quad u\in[h,1-h], \label{eq:T-conf-bands-iota}
\end{align}
where $c_{n,1-\alpha/2}$ is the $(1-\alpha/2)$-quantile of either of the random variables
\begin{align}
    W_n^{T,\iota} &\bydef \sup_{u^*\in[h,1-h]} \left| Z_n^{T*}(u^*)/\iota(u^*) \right|, \\
    W_n^{T,U[0,1],\iota} &\bydef \sup_{u^*\in[h,1-h]} \left| Z_n^{T,U[0,1]}(u^*)/\iota(u^*) \right|, 
\end{align}
These bands stay asymptotically exact for any $\iota$, but have the shape $\iota(u^*) \hat q_h(u^*)$, which may affect their finite-sample performance and asymptotic power.\footnote{\citet{montiel2019simultaneous} discuss a similar issue with simultaneous confidence bands for a \emph{vector} (rather than functional) parameter.}
\end{rem}



%% file: MC.tex
\section{Monte Carlo experiments} \label{sec:MC}

While our theoretical results establish the asymptotic validity of the confidence bands, they do not rule out a substantial finite-sample size distortion. In this section, we evaluate the extent of this distortion in a set of Monte Carlo experiments.

For simplicity, we simulate an auction with exactly two bidders ($M = 2$ and $p_2 = \check p_2 = 1$) and a non-binding (original) reserve price $\underline r = 0$.
We consider three choices for the distributions of the observed bids: uniform, beta and power-law; all of which are supported on the interval $[0,1]$.
The simplest choice is the uniform distribution, since it has a strictly positive density with $q(u) = 1$. The beta$(\alpha, \beta)$ distribution features a bell-shaped density for $\alpha, \beta > 1$, with varying skewness. In contrast, the density $f(x) = \alpha x^{\alpha-1}$ of the power-law distribution increases on the support $[0,1]$ for $\alpha > 1$. We censor these distributions at top 5\% and bottom 5\% quantile levels, so that the quantile density is strictly positive and satisfies the statement of \cref{prop:distr-of-bids}.\footnote{The censoring of the distribution in the Monte Carlo simulations is achieved by replacing their true quantile function $Q(u)$ with $(Q(0.05 + 0.9 u)-Q(0.05))/(Q(0.95)-Q(0.05))$. We emphasize that every simulated bid distribution can be rationalized by some value distribution satisfying \cref{ass:distr-bids}.}

\begin{table}[t!]
    \centering
    \begin{tabular}{l c c c c c}
        Estimand & (i) & (ii) & (iii) & (iv) & (v) \\
        \hline
        & \multicolumn{5}{c}{Sample size = 1,000,  trim = 3\%} \\
        \hline
beta(1,1)   & 0.95  & 0.952 & 0.912 & 0.91  & 0.974 \\
beta(2,2)   & 0.954 & 0.954 & 0.912 & 0.904 & 0.97  \\
beta(5,2)   & 0.952 & 0.954 & 0.924 & 0.916 & 0.966 \\
beta(2,5)   & 0.956 & 0.962 & 0.902 & 0.898 & 0.968 \\
powerlaw(2) & 0.952 & 0.952 & 0.928 & 0.922 & 0.976 \\
powerlaw(3) & 0.948 & 0.948 & 0.93  & 0.926 & 0.978 \\
        \hline
        & \multicolumn{5}{c}{Sample size = 10,000,  trim = 1.5\%} \\
        \hline
beta(1,1)   & 0.95  & 0.948 & 0.932 & 0.936 & 0.96  \\
beta(2,2)   & 0.954 & 0.954 & 0.932 & 0.934 & 0.96  \\
beta(5,2)   & 0.952 & 0.954 & 0.93  & 0.932 & 0.962 \\
beta(2,5)   & 0.952 & 0.952 & 0.918 & 0.93  & 0.958 \\
powerlaw(2) & 0.954 & 0.952 & 0.94  & 0.938 & 0.96  \\
powerlaw(3) & 0.948 & 0.952 & 0.934 & 0.938 & 0.96  \\
        \hline 
        & \multicolumn{5}{c}{Sample size = 100,000,  trim = .7\%} \\
        \hline
beta(1,1)   & 0.95  & 0.948 & 0.938 & 0.942 & 0.954 \\
beta(2,2)   & 0.952 & 0.948 & 0.944 & 0.946 & 0.956 \\
beta(5,2)   & 0.954 & 0.952 & 0.944 & 0.948 & 0.956 \\
beta(2,5)   & 0.956 & 0.952 & 0.932 & 0.948 & 0.954 \\
powerlaw(2) & 0.944 & 0.948 & 0.948 & 0.948 & 0.954 \\
powerlaw(3) & 0.946 & 0.948 & 0.952 & 0.95  & 0.952 \\
    \end{tabular}
    \caption{Simulated coverage of the 95\% uniform confidence bands.}
    \label{tab:coverage}
\end{table}

The estimation targets are (i) the bid quantile function $q$, (ii) the value quantile function $v$; and the following quantities as functions of the counterfactual reserve price: (iii) the potential bidder's expected surplus, (iv) the expected revenue, and (v) the total expected surplus, see \cref{table_one}.

For the non-counterfactual targets (i), (ii) and for the $T$-type counterfactuals (iii), (iv), we calculate the confidence bands by simulation from the left-hand side of the respective BK expansions under the uniform[0,1] bid distribution, see \cref{Thm:uni-bands-T}. For the $S$-type functional (v), the confidence bands are constructed by simulation from the estimated process
\begin{align}
     \hat{\mathcal{G}}(u^*) \bydef \frac{1}{\sqrt{n}} \sum_{i=1}^n \left( \hat f_{u^*}(U_i) - \E \, \hat f_{u^*}(U_i) \right),
\end{align}
where $U_i \sim \text{iid Uniform}[0,1]$ and $\hat f_{u^*}$ is equal to $f_{u^*}$ with the true values $q$ replaced by their estimates $\hat q_h$, see \cref{thm:S-unif-asy-distr}. We use the undersmoothing bandwidth $h = 1.06 s \cdot n^{-0.34}$, where $s$ is the standard deviation of bid spacings, and set both the number of DGP simulations and the number of simulations for the critical values to 500.

The results are shown in \cref{tab:coverage}. The simulated coverage can be seen to be close to the nominal level of $0.95$ for larger sample sizes, which validates our theoretical results in \cref{sec:BK,sec:func}.

%% file: empirical.tex
\section{Empirical application}
\label{sec:emp}

In this section, we apply our methodology to test the hypothesis about the optimality of the auction design employed in timber sales held by the U.S. Forest Service in between 1974 and 1989, see, e.g., \cite{haile2001auctions}. These auctions did not feature a reserve price (i.e., $\underline r=0$), which raises the question of whether the collected revenue could have been higher had the reserve price been set at a positive level. We use the data kindly provided by Phil Haile on his website.\footnote{\url{http://www.econ.yale.edu/~pah29/}}

We select a subsample of auctions that are sealed-bid and have at least 2 bidders, see \cref{fig:residuals}. As is common in the literature, we residualize the log-bids using available auction-level characteristics: year and location dummies, the logarithms of the tract advertised value and the Herfindahl index.\footnote{The exponentiated log-bid residuals (to which we will refer simply as \emph{bid residuals}) are interpreted as estimates of the idiosyncratic component of bids, while the exponentiated fitted values are interpreted as estimates of the common component of bids, see \cite{haile2003nonparametric}.} The latter is a measure of the homogeneity of the tract with respect to the timber species. This procedure is consistent with a multiplicative model of observed auction heterogeneity.\footnote{The asymptotic distributions of the test statistics are not affected by the error in the residualization procedure, as long as the estimates of the common component of bids converge at a faster (in our case parametric) rate, see \cite{haile2003nonparametric} and \cite{athey2007nonparametric}.} The distribution of bid residuals is truncated at the 5th percentile on each end, leaving 60758 observations, see \cref{fig:residuals}.

\begin{figure}[t!]
\centering
\includegraphics[width = .98 \linewidth]{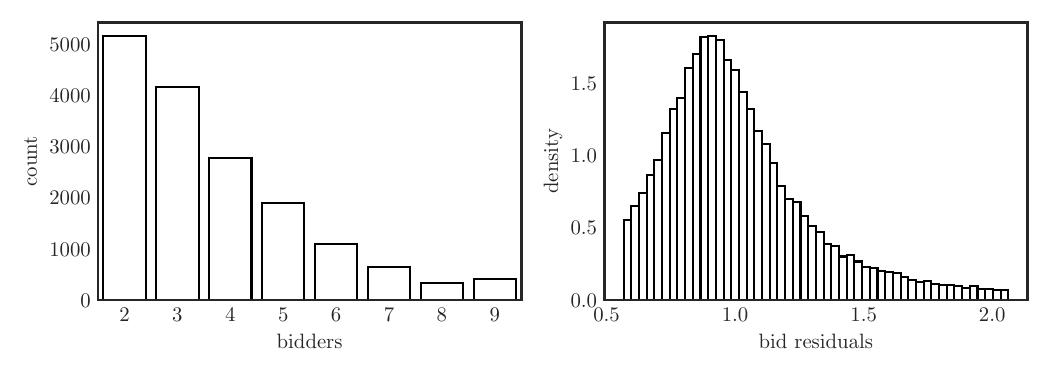}
\caption{Distributions of the number of bidders (left) and of bid residuals (right).}
\label{fig:residuals}
\end{figure}

To assess the change $\Delta(u^{\ast})$ in the expected revenue, we use the quantile version of the revenue formula in  \cref{table_one}, which yields
\begin{align*}
    \Delta(u^{\ast}) &\bydef \textit{RE}(u^*) - \textit{RE}(0) = \tilde M A_3(u^{\ast})v(u^{\ast}) - \int_{0}^{u^{\ast}} \left(A_2'(x) + \tilde M A_3'(x) \right) v(x) \, dx ,
\end{align*}
where $u^{\ast} > 0$ is the counterfactual exclusion level (rank of the counterfactual reserve price in the distribution of valuations) and we used the fact that in our application $v(0) = \underline r$. Since $\Delta(u^*)$ is similar to the $T$-type functional \eqref{eq:T-popul}, its estimator is
\begin{align}
    \hat \Delta_h(u^*) \bydef \varphi(u^*)\hat v_h(u^{\ast}) -\left[ \int_0^{u^*} \chi_\psi(u) \hat Q(u)\,du + \check A(u^*)\psi(u^*) \hat Q(u^*) - \check A(0)\psi(0) \hat Q(0) \right],
\end{align}
where $\varphi(x) \bydef \tilde M \check A_3(x)$, $\psi(x) \bydef A_2'(x) + \tilde M A_3'(x)$ and $\chi_\psi$ is defined in \eqref{eq:chi-def}.

We use the undersmoothing bandwidth $h = 1.06 s \cdot n^{-0.34}$, where $s$ is the standard deviation of spacings of bid residuals, and evaluate $\hat \Delta_h(u^{\ast})$ on the evenly spaced grid $\{i/n\}_{i=0}^n$. This bandwidth is slightly smaller than the Silverman rule of thumb bandwidth $h = 1.06 s \cdot  n^{-1/5}$.

To construct the confidence bands, we use the representation in \cref{thm:T-expansion} and \cref{thm:T-expansion-iota}. First, 1000 realizations of the bid quantile density $\hat q_{h}^{U}(\cdot)$ are simulated, independently from the data, based on pseudo-bids from the uniform[0,1] distribution. Second, for a nominal confidence level $(1-\alpha)$, the critical value $c_{n,1-\alpha}$ is computed as the $(1-\alpha)$-quantile of $\sup_{u \in [h,1-h]}(\hat q_{h}^{U}(u) - 1)$. Finally the one-sided confidence band is computed as
\begin{align}
\left(\hat\Delta_h(u) - \tilde M A_3(u) \check A(u) \cdot \hat q_h(u) \cdot  c_{n,1-\alpha}, \,\, +\infty \right), \quad u\in[h,1-h].    
\end{align}

\begin{figure}[t!]
\centering
\includegraphics[width = \linewidth]{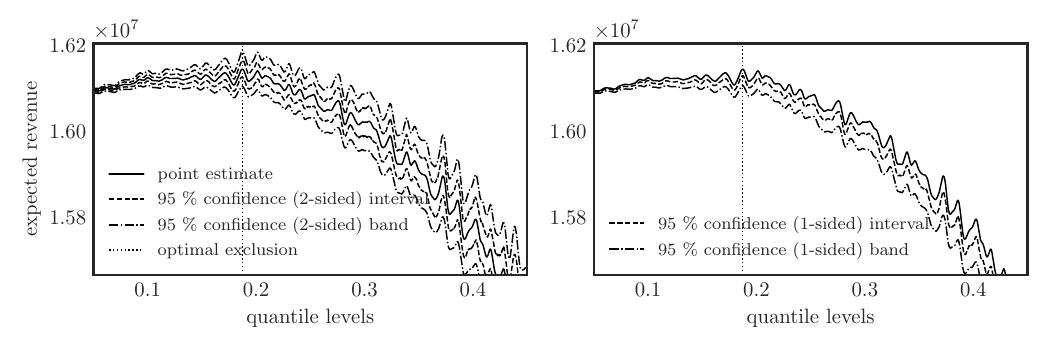}
\caption{Confidence intervals and bands for the counterfactual expected revenue.}
\label{fig:application}
\end{figure}

We test the hypothesis of the nonexistence of a counterfactual (positive) reserve price that would increase the seller's expected revenue. Formally, we test the hypotheses $H_0$ against $H_1$, where
\begin{equation}
    H_0: \sup_{u^{\ast} \in [h,1-h]} \Delta(u^{\ast}) = 0, \quad H_1: \sup_{u^{\ast} \in [h,1-h]} \Delta(u^{\ast}) > 0.
\end{equation}

The corresponding test statistic is the maximal (over the grid) value of the lower end point function of the one-sided confidence band, and $H_0$ is rejected whenever this maximum is positive. We denote by $\tilde u$ the point at which the maximum of the estimated function $\textit{RE}(u^*)$ is attained, i.e. the optimal exclusion level.


We test the hypothesis using subsamples of auctions with different numbers of bidders, see \cref{results}. We use subsamples with 2 and 3 bidders, and also with 2-5 (small auctions), 5-9 (large auctions), and 2-9 (all auctions) bidders.\footnote{Typically, a researcher would pick, for the sake of simplicity, a subsample of auctions with the same number of bidders. However, our methodology allows for a random number of bidders, so we can pool auctions with different numbers of bidders together.} Under all specifications, $H_0$ is rejected at a 95\% confidence level, 
see \cref{results}, meaning that the revenue gains at the optimal reserve price are statistically significant, albeit relatively small. We also point out that the sample size can be significantly increased by pooling together the auctions with a different number of bidders. For a full sample, the results of the estimation are illustrated in \cref{fig:application}.

\begin{table}[t!]
    \centering
    \begin{tabular}{l | c c c c c}
        Number of bidders & 2 & 3 & 2-5 & 5-9 & 2-9\\
        \hline
        \hline
sample size                &  10328 & 12477 & 43387 &  26841 &  60758\\
bandwidth relative to $[0,1]$               &  0.01 & 0.009 & 0.006 &  0.007 &  0.006\\
optimal exclusion $\tilde u$                  &  0.1 & 0.12 & 0.12 &  0.18 &  0.19\\
$H_0$ hypothesis   &  \text{reject} &  \text{reject} &  \text{reject} &  \text{reject} & \text{reject}
    \end{tabular}
    \caption{Test results at the 95\% confidence level.}
    \label{results}
\end{table}




%% file: practical.tex
\section{Practical considerations}
\label{sec:practical}

In this section, we briefly discuss some important technical aspects of our methodology. 

\paragraph{Choice of the grid.}
While it is theoretically possible to evaluate our estimators at any quantile level, choosing the evenly-spaced grid $\{i/n\}_{i=2}^n$ has a massive impact on the computational complexity of the estimation procedure and its performance.

Note that, with this grid, the estimate of $\hat q_h(u)$ becomes a discrete convolution of the vector of spacings $\{b_{(i)} - b_{(i-1)}\}_{i=2}^n$ with a discrete filter corresponding to $K_h$. The discrete convolution is a remarkably fast and reliable procedure. Moreover, the counterfactual estimators can be well approximated with the weighted cumulative sums of the vectors of spacings. Consequently, all our estimators can be thought of as combinations of elementary vector operations with sorting and convolution.



\paragraph{Shape restrictions.}
Since $v(\cdot)$ is a quantile function, one may want to impose monotonicity on $\hat v_h(\cdot)$ and the associated confidence bands. As suggested in \cite{chernozhukov2009improving,rearrangement}, an effective way of doing so is \emph{smooth rearrangement} of the estimate and the confidence bands, whose discrete counterpart is merely a sorting algorithm. We leave the analysis of such shape-restricted estimators for future work. We note that there is emerging literature exploiting shape restrictions for auction counterfactuals, e.g. \citet{henderson2012empirical,luo2018integrated,pinkse2019estimation,ma2021monotonicity}.



\paragraph{Alternative estimator of $q(u)$.}
A close competitor to our estimator $\hat q_h(\cdot)$ of the bid quantile density is the reciprocal of the kernel estimator $\hat f_l(\cdot)$ of the bid density, as in the first step of the procedure of \citet{guerre2000optimal},
\begin{equation*}
\tilde q_l\left( u \right) = \left(\hat f_l (b_{([un]+1)}) \right)^{-1} = \left( \frac{1}{nl}\sumin K\bigPar{\frac{b_{([un]+1)}-b_i}{l}} \right)^{-1}, \quad u \in [0,1].
\end{equation*}


An insightful comparison of $\hat q_h$ and $\tilde q_l$ was carried out by \cite{jones1992estimating} who showed that the variance components of the mean squared errors (MSE) of the estimators are equal if so-called \emph{scale match-up} bandwidths are used. Namely, for a fixed point $u \in (0,1)$, the MSE of $\hat q(u)$ is only less than or equal to that of $\tilde q_l(u)$ when $q(u) q''(u) \le 1.5 q'(u)^2$ or, equivalently,
$f(b) f''(b) \ge 1.5 f'(b)^2.$
Therefore, the reciprocal kernel density $\tilde q_l$ performs better close to the center of the distribution, while the kernel quantile density $\hat q_h$ is preferable at the tails.

Finally, we note that the algorithms to construct the estimates of $\tilde q_l$ and $\hat q_h$ on the grid, seem to have different computational complexity, which can be heuristically shown to be $O(n^2)$ and $O(n \log n)$, respectively. This is due to the fact that the convolution algorithm, which the estimator $\hat q_l(u)$ relies on, has the complexity of roughly $O(n \log n)$ due to its usage of the fast Fourier transform.



%% file: conclusion.tex
\section{Conclusion}
\label{sec:conclude}

In this paper, we develop a novel approach to estimation and inference on counterfactual functionals of interest (such as the expected revenue as a function of the rank of the counterfactual reserve price) in the standard nonparametric model of the first-price sealed-bid auction.
We show that these counterfactuals can be written as continuous linear functionals of the quantile function of bidders' valuations, which can be recovered from the observed bids using a well-known explicit formula.
We suggest natural estimators of the counterfactuals and show that their asymptotic behavior depends on their structure.
In particular, we classify the counterfactuals into two types, one allowing for parametric (fast) convergence rates and standard inference, and the other exhibiting nonparametric (slow) convergence rates and lack of uniform convergence. 
For each of the types of counterfactuals, we develop simple, simulation-based algorithms for constructing pointwise confidence intervals and uniform confidence bands.

In terms of efficiency, our estimators of the ``less smooth'' targets, such as the bidder's surplus and auctioneer's revenue, perform on par with the best that the literature can offer, achieving the optimal rate of estimation while using undersmoothing for inference. However, there may be subtle differences in implementation. Indeed, our approach constructs the estimator on the grid of quantiles simultaneously (rather than separately) for each counterfactual reserve price.
Combined with the Fast Fourier Transform for discrete convolution, this leads to significant computational gains.
The latter can only be seen in massive simulations, however, so the final choice of the estimator should be ultimately based on specific goals and computational resources.

Avenues for further research include boundary correction, data-driven bandwidth selection, shape-restricted estimation of the valuation quantile function and associated functionals, and auction heterogeneity.

\section*{Acknowledgements}
We are grateful to Karun Adusumilli, Tim Armstrong, John Asker, Yuehao Bai, Zheng Fang, Sergio Firpo, Antonio Galvao, Andreas Hagemann, Vitalijs Jascisens, Hiroaki Kaido, Nail Kashaev, Michael Leung, Tong Li, Vadim Marmer, Hyungsik Roger Moon, Hashem Pesaran, Guillaume Pouliot, Geert Ridder, Pedro Sant'Anna, Yuya Sasaki, Matthew Shum, Liyang Sun, Takuya Ura, Quang Vuong, Kaspar W\"{u}thrich, and seminar participants at USC for valuable comments. All errors and omissions are our own.

%% file: appendix.tex
\section{Estimation and inference for value quantiles, proofs}

\subsection{Proof of \cref{Thm:BK-expansion}}

First, we need the following two lemmas concerning expressions that appear further in the proof.

\begin{lem}\label{Lem:int-by-parts-kqe}
Suppose $K$ is a continuous function of bounded variation. Then
\begin{align}
    \int_0^1 K_h(u-z) \, d\left(\hat Q(z)-Q(z)\right) = - \int_0^1 \left(\hat Q(z)-Q(z)\right)  \, d K_h(u-z) + R^I_n(u),
\end{align}
where $\sup_{u\in [0,1]} |R^I_n(u)| = O_{a.s.}\left(\frac{1}{nh}\right)$.
\end{lem}

\begin{proof}
Denote $\hat\psi(z)=\hat Q(z)-Q(z)$ and note that $\hat \psi$ is a function of bounded variation a.s. Using integration by parts for the Riemann-Stieltjes integral \citep[see e.g.][Theorem 1.2.7]{stroock1998concise}, we have
\begin{align}
    \int_0^1 K_h(u-z) \, d\hat \psi(z) = - \int_0^1 \hat\psi(z)  \, d K_h(u-z) + K_h(u-1) \hat\psi(1) - K_h(u)\hat\psi(0)
\end{align}
To complete the proof, note that $\hat\psi(1)=b_{(n)}-\bar b = O_{a.s.}(n^{-1})$, $\hat\psi(0)=b_{(1)}-\underline b = O_{a.s.}(n^{-1})$, $|K_h(u-1)|\le h^{-1}K(0)$ and $|K_h(u)|\le h^{-1}K(0)$.
\end{proof}

\begin{lem}\label{Lem:int-by-parts-kde}
Suppose $K$ is a continuous function of bounded variation. Then, for every $u \in [0,1]$,
\begin{align}
    &Z_n^*(u) \bydef \sqrt{nh} \int_0^1 (\hat F(Q(z))-z) \, d K_h(u-z) = -\eG_{n,h}(u),\\
    &\eG_{n,h}(u) \bydef \frac{\sqrt{nh}}{n} \sumin \left[K_h(u-F(b_i)) - \E K_h(u-F(b_i)) \right].
\end{align}
\end{lem}
\begin{proof}
Using integration by parts for the Riemann-Stieltjes integral \citep[see e.g.][Theorem 1.2.7]{stroock1998concise}, we have
\small
\begin{align*}
    \int_0^1 (\hat F(Q(z))-z) \, d K_h(u-z) &= -\int_0^1  K_h(u-z) \, d \left[\hat F(Q(z))-z\right] + K_h(u-1) \left[ \hat F(\bar b)-1 \right] + K_h(u)\hat F(0) \\
    &= -\int_0^1  K_h(u-z) \, d \left[\hat F(Q(z))-z\right],
\end{align*}
\normalsize
where we used the fact that $\hat F(\bar b) = 1$ a.s. and $\hat F(0)=0$ a.s. We further write
\begin{align*}
        \int_0^1(\hat F(Q(z))-z) \, d K_h(u-z) &= -\int_0^1 K_h(u-z) \, d \left[\hat F(Q(z))-z\right] \\
        &= -\int_0^{\bar b} K_h(u-F(x)) \, d \left[\hat F(x)-F(x) \right] \\
        &= -\frac{1}{n} \sumin \left[K_h(u-F(b_i)) - \E K_h(u-F(b_i)) \right],
\end{align*}
where in the second equality we used the change of variables $x=Q(z)$.
\end{proof}


We now proceed with the proof of Theorem \ref{Thm:BK-expansion}.

Plug in the BK expansion \eqref{E:BK-general} and use Lemma \ref{Lem:int-by-parts-kqe} to obtain
\begin{align}
\hat q_h(u) - q_h(u) &= \int_0^1 K_h(u-z) \, d\left[ \hat Q(z)-Q(z)\right] \\
&= \int_0^1 \left[ \hat Q(z)-Q(z)\right] \, d K_h(u - z) + R_n^I(u) \\
&= \int_0^1  q(z) (\hat F(Q(z))-z)\, d K_h(u - z) + \int_0^1  R_n^{BK}(z) \, d K_h(u - z) + R_n^I(u). \label{E:qh-q}
\end{align}

\textbf{First term in \eqref{E:qh-q}}.

Since $f$ is bounded away from zero, $|q'|\le M< \infty$ for some constant $M$, and hence $|q(z)-q(u)|\le M|z-u|$. The first term in \eqref{E:qh-q} can then be rewritten as 
\begin{align}
        \int_0^1q(z) (\hat F(Q(z))-z) \, d K_h(u-z) &= q(u) \int_0^1(\hat F(Q(z))-z) \, d K_h(u-z) + R^{II}_n(u), \label{E:first-term-split}
\end{align}
where
\begin{align}
    \left|R_n^{II}(u)\right| &= \left| \int_0^1(q(z)-q(u))  (\hat F(Q(z))-z)\, d K_h(u - z) \right| \\
    &\le Mh \left| \int_0^1(\hat F(Q(z))-z) \, d K_h(u-z) \right| = Mh \left| (nh)^{-1/2} Z_n^*(u) \right|.
\end{align}
By Lemma \ref{Lem:int-by-parts-kde}, $Z_n^*(u)=-\eG_{n,h}(u)$, where the process $\eG_{n,h}(u) = O_{a.s.}(\log h)$ uniformly in $u \in [0,1]$ \citep[see e.g.][]{silverman1978weak,stute1984oscillation}, and hence
\begin{align}
    R_n^{II}(u) = O_{a.s.}\left(\frac{h\log h}{\sqrt{nh}}\right) \text{ uniformly over } u \in (0,1).
\end{align}
Applying Lemma \ref{Lem:int-by-parts-kde} to the first term in \eqref{E:first-term-split} allows us to rewrite
\begin{align}
    \int_0^1 q(z) (\hat F(Q(z))-z) \, d K_h(u-z) &= -q(u) (nh)^{-1/2} \eG_{n,h}(u) + O_{a.s.}\left(\frac{h\log h}{\sqrt{nh}}\right). \label{E:first-term-rate}
\end{align}

\textbf{Second term in \eqref{E:qh-q}}.

This term can be upper bounded as follows,
\begin{align}
        &\sup_u \left|\int_0^1R_n^{BK}(z) \, d K_h(u-z) \right| \le \sup_u \int_0^1\left|  R_n^{BK}(z) \right| \,\left| d K_h(u-z)\right| \\
        &\le \sup_z |R_n^{BK}(z)| TV(K_h) = O_{a.s.}\left(n^{-3/4}\ell(n)\right) h^{-1}TV(K) =
         O_{a.s.}\left(h^{-1}n^{-3/4}l(n)\right), \label{E:second-term-rate}
\end{align}
where we used the properties of total variation in the first inequality and in the second equality.

Plugging \eqref{E:first-term-rate} and \eqref{E:second-term-rate} into \eqref{E:qh-q} and multiplying by $\sqrt{nh}$ yields
\begin{align}
    \sqrt{nh}\left(\hat q_h(u)-q_h(u) \right) = -q(u) \eG_{n,h}(u) + O_{a.s.}\left(h\log h + h^{-1/2}n^{-1/4}\ell(n)\right), \label{eq:qhat-q}
\end{align}
where we disregarded the term $\sqrt{nh} R_n^I(u)$, since it has the uniform order $O_{a.s.}(n^{-1/2}h^{-1/2})$, which is smaller than  $O_{a.s.}\left(h^{-1/2}n^{-1/4}\ell(n)\right)$.

Note that, for $u \in [h,1-h]$, there exists $\zeta(u,z)$ lying between $u$ and $z$ such that
\begin{align}
    q_h(u) &= \int_0^1 q(z) K_h(u-z) \, dz = \int_0^1 \left( q(u) + q'(\zeta(u,z))(u-z) \right)K_h(u-z)\, dz \\
    &= q(u) + O(h).
\end{align}
Combining this with \eqref{eq:qhat-q} yields
\begin{align}
    &\sqrt{nh}\left(\hat q_h(u)-q(u) \right) = -q(u) \eG_{n,h}(u) + O_{a.s.}\left(n^{1/2} h^{3/2} + h\log h + h^{-1/2}n^{-1/4}\ell(n)\right), \label{eq:qhat-q-infeas}
\end{align}
uniformly in $u \in [h,1-h]$. Using $\eG_{n,h}(u) = O_{a.s.}(\log h) $ again, we conclude that
\begin{align}
    \sqrt{nh}\left(\hat q_h(u)-q(u) \right) = O_{a.s.}\left(\log h + n^{1/2} h^{3/2}\right)
\end{align}
(note that we dropped the terms $h \log h$ and $h^{-1/2}n^{-1/4}\ell(n)$ since they are smaller than $\log h$) or, dividing by $\sqrt{nh}$,
\begin{align}
    \hat q_h(u)-q(u) = O_{a.s.}\left(\frac{\log h}{\sqrt{nh}} + h\right), \text{ uniformly in } u \in [h,1-h]. \label{eq:qh-unif-consistency}
\end{align}

Now we replace $q(u) \eG_{n,h}(u)$ by $\hat q_h(u) \eG_{n,h}(u)$ in \eqref{eq:qhat-q-infeas}, which leads to the approximation error
\begin{align}
    \eG_{n,h}(u) \left( \hat q_h(u)-q(u) \right) = O_{a.s} \left(\frac{(\log h)^2}{\sqrt{nh}} + h\log h\right), \text{ uniformly in } u \in [h,1-h],
\end{align}
as follows from \eqref{eq:qh-unif-consistency}.
Hence, \eqref{eq:qhat-q-infeas} becomes
\begin{align}
    \sqrt{nh}\left(\hat q_h(u)-q(u) \right) = - \hat q_h(u) \eG_{n,h}(u) +  O_{a.s.}\left(n^{1/2}h^{3/2} + h\log h + h^{-1/2}n^{-1/4}\ell(n)\right), \label{eq:qhat-q-feas}
\end{align}
uniformly in $u \in [h,1-h]$, where we dropped the term $\frac{(\log h)^2}{\sqrt{nh}}$ since it is smaller than $h^{-1/2}n^{-1/4}\ell(n)$.

Finally, write
\begin{align}
\sqrt{nh}\left(\hat v_h(u)-v(u)\right) &= \sqrt{nh} \left(\hat Q(u)-Q(u)\right) + \check A(u)\sqrt{nh} \left(\hat q_h(u)-q(u) \right).\label{eq:vhat-v}
\end{align}
Since $\hat Q(u)-Q(u)=O_{a.s.}(n^{-1/2})$ uniformly in $u\in(0,1)$, we have
\begin{align}
    \sqrt{nh}\left(\hat v_h(u)-v(u)\right) = O_{a.s.}(h^{1/2}) + \check A(u) \sqrt{nh}\left(\hat q_h(u)-q(u) \right) .
\end{align}
Combining this with \eqref{eq:qhat-q-feas} and noting that $h\log h$ is smaller than $h^{1/2}$ yields
\begin{align}
    \sqrt{nh}\left(\hat v_h(u)-v(u)\right) = - \check A(u) \hat q_h(u) \eG_{n,h}(u) +  O_{a.s.}\left(n^{1/2}h^{3/2} + h^{1/2} + h^{-1/2}n^{-1/4}\ell(n)\right),
\end{align}
uniformly in $u\in[h,1-h]$. Dividing by $\hat q_h(u)$, which is bounded away from zero w.p.a. 1, completes the proof. $\qed$

\subsection{Proof of \cref{Thm:uni-conf-bands}}

A key ingredient of the proof is to note that Lemmas 2.3 and 2.4 of \citet{chernozhukov2014gaussian} continue to hold even if their random variable $Z_n$ does not have the form $Z_n=\sup_{f\in \mathcal{F}_n} \eG_nf$ for the standard empirical process $\eG_n$, but instead is a generic random variable admitting a strong sup-Gaussian approximation with a sufficiently small remainder.

For completeness, we provide the aforementioned trivial extensions of the two lemmas here.

Let $X$ be a random variable with distribution $P$ taking values in a measurable space $(S,\mathcal{S})$. Let $\F$ be a class of real-valued functions on $S$. We say that a function $F: S \to \R$ is an \emph{envelope} of $\F$ if $F$ is measurable and $|f(x)|\le F(x)$ for all $f \in \F$ and $x \in S$.

We impose the following assumptions (A1)-(A3) of \citet{chernozhukov2014gaussian}.
\begin{enumerate}
    \item[(A1)] The class $\F$ is \emph{pointwise measurable}, i.e. it contains a coutable subset $\G$ such that for every $f\in \F$ there exists a sequence $g_m \in \G$ with $g_m(x) \to f(x)$ for every $x \in S$.
    \item[(A2)] For some $q\ge 2$, an envelope $F$ of $\F$ satisfies $F \in L^q(P)$.
    \item[(A3)] The class $\F$ is $P$-pre-Gaussian, i.e. there exists a tight Gaussian random variable $G_P$ in $l^\infty(\F)$ with mean zero and covariance function
    \begin{align*}
        \E[G_P(f)G_P(g)] = \E[f(X)g(X)] \text{ for all } f,g\in\F.
    \end{align*}
\end{enumerate}

\begin{lem}[A trivial extension of Lemma 2.3 of \citet{chernozhukov2014gaussian}] \label{Lem:CCK-2-3}
Suppose that Assumptions (A1)-(A3) are satisfied and that there exist constants $\underline \sigma$, $\bar \sigma>0$ such that $\underline\sigma^2 \le Pf^2 \le \bar\sigma^2$ for all $f\in\F$. Moreover, suppose there exist constants $r_1,r_2>0$ and a random variable $\tilde Z=\sup_{f\in \F} G_Pf$ such that $\Pb(|Z-\tilde Z|>  r_1)\le r_2$. Then
\begin{align*}
    \sup_{t\in\R}\left|\Pb(Z\le t)-\Pb(\tilde Z\le t)\right| \le C_{\sigma} r_1\left\{ \E \tilde Z + \sqrt{1 \vee \log(\underline\sigma/r_1)} \right\} + r_2,
\end{align*}
where $C_\sigma$ is a constant depending only on $\underline\sigma$ and $\bar\sigma$.
\end{lem}

\begin{proof}
For every $t\in \R$, we have
\begin{align*}
    \Pb(Z\le t) &= \Pb(\{Z\le t\} \cap \{|Z-\tilde Z|\le r_1\}) + \Pb(\{Z\le t\} \cap \{|Z-\tilde Z|>r_1\})\\
    &\le \Pb(\tilde Z\le t+r_1)+r_2\\
    &\le \Pb(\tilde Z\le t) + C_{\sigma} r_1\left\{ \E \tilde Z + \sqrt{1 \vee \log(\underline\sigma/r_1)} \right\} + r_2,
\end{align*}
where Lemma A.1 of \citet{chernozhukov2014gaussian} (an anti-concentration inequality for $\tilde Z$) is used to deduce the last inequality. A similar argument leads to the reverse inequality, which completes the proof.
\end{proof}

\begin{lem}[A trivial extension of Lemma 2.4 of \citet{chernozhukov2014gaussian}] \label{Lem:CCK-2-4}
Suppose that there exists a sequence of $P$-centered classes $\F_n$ of measurable functions $S \to \R$ satisfying assumptions (A1)-(A3) with $\F=\F_n$ for each $n$, where in the assumption (A3) the constants $\underline\sigma$ and $\bar\sigma$ do not depend on $n$. Denote by $B_n$ the Brownian bridge on $\ell^\infty(\F_n)$, i.e. a tight Gaussian random variable in $\ell^\infty(\F_n)$ with mean zero and covariance function
\begin{align*}
    \E[B_n(f)B_n(g)] = \E[f(X)g(X)] \text{ for all } f,g\in\F_n.
\end{align*}
Moreover, suppose that there exists a sequence of random variables $\tilde Z_n = \sup_{f\in\F_n} B_n(f)$ and a sequence of constants $r_n\to 0$ such that $|Z_n-\tilde Z_n|=O_P(r_n)$ and $r_n \E \tilde Z_n \to 0$. Then 
\begin{align*}
    \sup_{t\in \R} \left|\Pb(Z_n\le t)-\Pb(\tilde Z_n\le t)\right| \to 0.
\end{align*}
\end{lem}

\begin{proof}
Take $\beta_n\to\infty$ sufficiently slowly such that $\beta_nr_n(1\vee \E\tilde Z_n)=o(1)$. Then since $\Pb(|Z_n-\tilde Z_n|>\beta_n r_n)=o(1)$, by Lemma \ref{Lem:CCK-2-3}, we have
\begin{align*}
    \sup_{t\in \R} \left|\Pb(Z_n\le t)-\Pb(\tilde Z_n\le t)\right| = O\left( r_n(\E\tilde Z_n + |\log (\beta_n r_n)|) \right) + o(1) = o(1).
\end{align*}
This completes the proof.
\end{proof}

\begin{lem}\label{lem:gauss-aprox-w}
Let 
\begin{align}
W_{n}^* = \sup_{u \in [0,1]} \xi(u) \sqrt{nh} \cdot \frac{1}{n} \sumin \left[ K_{h}(U_i-u) - \E K_{h}(U_i-u) \right]
\end{align}
for some smooth function $\xi:\,[0,1]\to\R$. Then there exists a tight centered Gaussian random variable $B_n$ in $\ell^\infty([0,1])$ with the covariance function
\begin{align}
\E[B_n(u)B_n(v)] = \xi(u)\xi(v) \cdot \text{Cov}\left(K_{h}(U-u), K_{h}(U-v) \right), \quad u,v\in[0,1],
\end{align}
such that, for $\tilde W_n=\sup_{u\in[0,1]} B_n(u)$, we have the approximation
\begin{align}
    W_n^* = \tilde W_n + O_{p}\left((nh)^{-1/6}\log n\right) \label{E:Gauss-approx-linear-term}.
\end{align}
\end{lem}

\begin{proof}
Define the class of functions 
\begin{align}
\mathcal{F}_n=\{[0,1] \ni x\mapsto \xi(u) K_{h}(u-x), \quad u\in[0,1]\}
\end{align}
and note that
\begin{align}
W_n^*=\sqrt{h}\|\G_n\|_{\mathcal{F}_n}.
\end{align}
Let us apply \citet[][Proposition 3.1]{chernozhukov2014gaussian} to obtain a sup-Gaussian approximation of $W_n^*$. Indeed, in the notation of \citet[][Section 3.1]{chernozhukov2014gaussian}, take
\begin{align}
g\equiv 1, \quad \mathcal{G}=\{g\}, \quad \mathcal{I} = [0,1], \quad c_n(u,g)=\xi(u).
\end{align}
Then the representation (8) in \citet{chernozhukov2014gaussian} holds, i.e.
\begin{align}
W_n^* = \sup_{(u,g) \in \mathcal{I}\times \mathcal{G}} c_n(u,g) \sqrt{nh} \cdot \frac{1}{n} \sumin \left[ K_{h}(U_i-u) - \E K_{h}(U_i-u) \right].
\end{align}
It is now trivial to check that the assumptions of \citet[][Proposition 3.1]{chernozhukov2014gaussian} hold and the statement of the lemma follows.
\end{proof}

Let us now go back to the proof of \cref{Thm:uni-conf-bands}. Use \cref{lem:gauss-aprox-w} with $\xi(u)=u$ and note that  \cref{Lem:CCK-2-4} and \citet[][Remark 3.2]{chernozhukov2014gaussian} then imply
\begin{align}
    \sup_{t\in \R} \left|\Pb(W_n^*\le t)-\Pb(\tilde W_n\le t)\right| \to 0. \label{E:Kolm-conv-1}
\end{align}
On the other hand, by \cref{Thm:BK-expansion} we have
\begin{align}
    W_n = W_n^* + O_{a.s.}\left( h^{1/2} + h^{-1/2}n^{-1/4}l(n)\right).
\end{align}
Substituting \eqref{E:Gauss-approx-linear-term} into this equation, we obtain
\begin{align}
    W_n = \tilde W_n + O_{p}\left((nh)^{-1/6}\log n + h^{1/2} + h^{-1/2}n^{-1/4}l(n)\right).
\end{align}
\cref{ass:band-small} then implies $W_n-\tilde W_n=o_p(\log^{-1/2} n)$. \citet[][Remark 3.2]{chernozhukov2014gaussian} now implies
\begin{align}
    \sup_{t\in \R} \left|\Pb(W_n\le t)-\Pb(\tilde W_n\le t)\right| \to 0. \label{E:Kolm-conv-2}
\end{align}

Given \eqref{E:Kolm-conv-1} and \eqref{E:Kolm-conv-2}, applying the triangle inequality finishes the proof. $\qed$


\section{Estimation and inference for counterfactuals, proofs}

\subsection{Proof of \cref{thm:S-unif-asy-distr}}

First, write
\begin{align}
\hat S(u^*) - S(u^*) &= \int_{u^*}^1 \varphi(u) (\hat Q(u)-Q(u))\,du \\
&- \check A(u^*)\psi(u^*)(\hat Q(u^*)-Q(u^*)) + \check A(1)\psi(1)(\hat Q(1)-Q(1)).
\end{align}
Using the classical BK expansion \eqref{E:BK-general}, we obtain
\begin{align}
\hat S(u^*) - S(u^*) &= - \int_{u^*}^1 \varphi(u)q(u)\left[\hat F(Q(u)) - u\right]\,du \\
&+ \check A(u^*) \psi(u^*)q(u^*) \left[\hat F(Q(u^*)) - u^*\right] + R_n(u^*),
\end{align}
where the composite error term
\begin{align}
R_n(u^*) &= \int_{u^*}^1 \varphi(u)r_n(u)\,du - \check A(u^*)\psi(u^*)r_n(u^*) + \check A(1)\psi(1) (\hat Q(1)-Q(1)) \\
&= O_{a.s.}\left( n^{-3/4}\ell(n)\right),
\end{align}
uniformly in $u^* \in [0,1]$. The latter rate follows from \eqref{E:BK-general} and the fact that $\hat Q(1)-Q(1) = O_{a.s.}(n^{-1})$.

Denoting $U_i=F(b_i)$, we can write
\begin{align}
\sqrt{n}\left(\hat S(u^*) - S(u^*) \right) &= \frac{1}{\sqrt{n}}\sumin\left[f_{u^*}(U_i)-\E f_{u^*}(U_i) \right] +  O_{a.s.}\left( n^{-1/4}\ell(n)\right).
\end{align}

Let us show that the class $\{f_{u^*}\,\vert\,u^*\in[0,1]\}$ is Donsker.

Since the sum of a finite number of Donsker classes is Donsker \citep[see][]{alexander1987central}, and also a constant times a Donsker class is Donsker, it suffices to show that the following (uniformly bounded) classes are Donsker,
\begin{align}
\mathcal{H} &= \{h_{u^*}: U\mapsto \int_{u^*}^1 \varphi(u)q(u) 1(U\le u)\,du\,\vert\, u^*\in[0,1]\},\\
\mathcal{G} &= \{g_{u^*}: U\mapsto \check A(u^*)\psi(u^*) q(u^*) 1(U\le u^*) \,\vert\, u^*\in[0,1]\}.
\end{align}

Indeed, for $u^*,v^*\in[0,1]$,
\begin{align}
\left| h_{u^*}(U)-h_{v^*}(U)\right| &=\left| \int_{\min(u^*,v^*)}^{\max(u^*,v^*)} \varphi(x)q(x) 1(U\le x)\,dx \right| \le |u^*-v^*| \sup_{u\in[0,1]} |\varphi(u)q(u)|,
\end{align}
i.e. $\mathcal{H}$ is Lipschitz in parameter. Its Donskerness follows by Theorem 2.7.11 and Theorem 2.5.6 in \citet{vaart1996weak}.

On the other hand, $\mathcal{G}$ is Donsker since it is a product of the set of constant functions $U \mapsto \check A(u^*)\psi(u^*) q(u^*)$ (which is trivially Donsker) and the VC class $\{1(\cdot\le u^*),\,\,u^*\in[0,1]\}$. $\qed$

\subsection{Proof of \cref{thm:T-expansion}}

We have
\begin{align}
    &\sqrt{nh}\left( \hat T_h(u^*)-T(u^*) \right) = \check \varphi(u^*) \left(\hat v_h(u^*)-v(u^*)\right) + \sqrt{nh}\left( \hat S_{\check \psi}(u^*) - S_{\psi}(u^*) \right) \\
    &= \check \varphi(u^*)\check A(u^*) \hat q_h(u^*)\left( -\eG_{n,h}(u^*) + R_{n}(u) \right) + O_p\left( h^{1/2}\right)\\
    &=-\check \varphi(u^*)\check A(u^*) \hat q_h(u^*) \eG_{n,h}(u^*) + O_p\left( n^{1/2}h^{3/2} + h^{1/2} + h^{-1/2}n^{-1/4}\ell(n) \right),
\end{align}
uniformly in $u^*\in[h,1-h]$, where the last two equations use \cref{Thm:BK-expansion} and \cref{thm:S-unif-asy-distr}. Dividing by $\hat q_h(u^*)$, which is bounded away from zero w.p.a. 1, finishes the proof. $\qed$

\subsection{Proof of \cref{Thm:uni-bands-T}}

Use \cref{lem:gauss-aprox-w} with $\xi(u)=u\varphi(u)$ and note that Lemma \ref{Lem:CCK-2-4} and Remark 3.2 in \citet{chernozhukov2014gaussian} then imply
\begin{align}
    \sup_{t\in \R} \left|\Pb(W_n^{T*}\le t)-\Pb(\tilde W_n^T \le t)\right| \to 0. \label{E:Kolm-conv-1-T}
\end{align}
On the other hand, by \cref{thm:T-expansion} we have
\begin{align}
    W_n = W_n^* + O_{a.s.}\left( h^{1/2} + h^{-1/2}n^{-1/4}l(n)\right).
\end{align}
Substituting \eqref{E:Gauss-approx-linear-term} into this equation, we obtain
\begin{align}
    W_n^T = \tilde W_n^T + O_{p}\left((nh)^{-1/6}\log n + h^{1/2} + h^{-1/2}n^{-1/4}l(n)\right).
\end{align}
Under the assumption of the theorem, $h$ decays polynomially, and hence $W_n^T-\tilde W_n^T=o_p(\log^{-1/2} n)$. Remark 3.2 of \citet{chernozhukov2014gaussian} now implies
\begin{align}
    \sup_{t\in \R} \left|\Pb(W_n^T\le t)-\Pb(\tilde W_n^T\le t)\right| \to 0. \label{E:Kolm-conv-2-T}
\end{align}

Given \eqref{E:Kolm-conv-1-T} and \eqref{E:Kolm-conv-2-T}, applying the triangle inequality finishes the proof. $\qed$